\documentclass[11pt]{article}

\usepackage[usenames,dvipsnames]{xcolor}
\usepackage[colorlinks,citecolor=blue,linkcolor=BrickRed]{hyperref}
\usepackage{makeidx} 
\usepackage{algorithm}
\usepackage{algorithmic}
\usepackage{graphicx,tipa}
\usepackage{arcs,lmodern,fix-cm}
\usepackage{times}
\usepackage{amsfonts,latexsym,graphicx,epsfig,amssymb,color}
\usepackage{mathdots,amsmath,amsthm,amstext,setspace}
\usepackage{amsmath,amstext,amsthm,url,setspace, enumerate}
\usepackage{slashbox,multirow}
\usepackage{rotating}
\usepackage{verbatim}
\newtheorem{theorem}{Theorem}[section]

\newtheorem{lemma}[theorem]{Lemma}

\theoremstyle{definition}

\newtheorem{definition}[theorem]{Definition}

\renewenvironment{proof} {\noindent\textbf{Proof:} } { \qed \medskip}

\newcommand{\reals}{\mathbb{R}}

\newcommand{\ignore}[1]{}
\newcommand{\queen}{\mathcal Q}
\newcommand{\ser}[1]{\mathcal S_{#1}}

\newcommand{\norm}[1]{\left\lVert#1\right\rVert}
\newcommand{\sinn}[1]{\sin \left({#1}\right)}
\newcommand{\coss}[1]{\cos \left({#1}\right)}

\newcommand{\ki}[1]{C_{#1}}
\newcommand{\ci}[1]{\mathcal{C}(#1)}
\newcommand{\li}[1]{\mathcal{L}(#1)}
\newcommand{\pe}[1]{\textsc{PE}$_{#1}$}
\newcommand{\interval}[1]{\mathbb{I}{(#1)}}
\newcommand{\sol}[1]{\mathbb{S}_{#1}}
\newcommand{\acoss}[1]{\cos^{-1} \left({#1}\right)}

\newcommand{\diff}[2]{\frac{d{#1}}{d{#2}}}
\newcommand{\eps}{\epsilon}
\newcommand{\includeFig}[4]{\begin{figure}[htb!] \begin{center} \includegraphics[{#1}]{#3}\caption{\label{#2}#4} \end{center} \end{figure}}
\newcommand{\ID}{\mathcal{D}}

\newcommand{\UC}{\mathcal{U}}

\newcommand{\IT}{\mathcal{T}}

\newcommand{\pair}[2]{\left({#1},\  {#2}\right)}

\begin{document}
\title{\bf 
God Save the Queen
\footnote{This is the full version of the paper with the same title which will appear in the proceedings of the 
9th International Conference on Fun with Algorithms, (FUN'18), June 13--15, 2018, La Maddalena, Maddalena Islands, Italy.
}
}

\author{
Jurek Czyzowicz\footnotemark[1] ~\footnotemark[5] 
\and
Konstantinos Georgiou\footnotemark[2] ~\footnotemark[5] 
\and
Ryan Killick\footnotemark[4] ~\footnotemark[6]
\and 
Evangelos Kranakis\footnotemark[4] ~\footnotemark[5]
\and
Danny Krizanc\footnotemark[7]
\and
Lata Narayanan\footnotemark[8] ~\footnotemark[5] 
\and
Jaroslav Opatrny\footnotemark[8]
\and
Sunil Shende\footnotemark[9]
}

\def\thefootnote{\fnsymbol{footnote}}
\footnotetext[5]{Research supported in part by NSERC of Canada.}
\footnotetext[6]{Research supported in part by the Ontario Graduate Scholarship (OGS) Program.}
\footnotetext[1]{
Universite du Québec en Outaouais, Gatineau, Qu\'{e}bec, Canada, \texttt{jurek.czyzowicz@uqo.ca}
}
\footnotetext[2]{
Dept. of Mathematics, 
Ryerson University, 
Toronto, ON, Canada, \texttt{konstantinos@ryerson.ca}
}
\footnotetext[4]{
School of Computer Science, Carleton University, Ottawa ON, Canada, \texttt{ryankillick,kranakis@scs.carleton.ca}
}
\footnotetext[7]{
Department of Mathematics \& Comp. Sci., Wesleyan University, Middletown, CT, USA, \texttt{dkrizanc@wesleyan.edu}
}
\footnotetext[8]{
Department of Comp. Sci. and Software Eng., Concordia University, Montreal, Qu\'{e}bec,  Canada, \texttt{lata,opatrny@encs.concordia.ca}
}
\footnotetext[9]{
Department of Computer Science, Rutgers University, Camden, USA,
\texttt{shende@camden.rutgers.edu}
}
\maketitle

\begin{abstract}
Queen Daniela of Sardinia is asleep at the center of  a round room at the top of the tower in her castle. She is accompanied by her faithful servant, Eva. Suddenly, they are
awakened by cries of ``Fire''. The room is pitch black and they are disoriented. There is exactly one exit from the room somewhere along its boundary. They
must find it as quickly as possible in order to save the life of the queen. It is known that with two people searching while moving at maximum speed 1 anywhere in the room, the room can be evacuated (i.e., with
both people exiting) in $1 + \frac{2\pi}{3} + \sqrt{3} \approx 4.8264$ time units and this is optimal
%~[Czyzowicz et al., DISC'14],
\cite{CGGKMP}, 
assuming that the first person to find the exit can directly guide the other
person to the exit using her voice. 
Somewhat surprisingly, in this paper we show that if the goal is to save the queen (possibly leaving Eva behind to die in
the fire) there is a slightly better strategy. We prove that this ``priority'' version of evacuation can be solved in time at most $4.81854$. Furthermore, we show that any
strategy for saving the queen requires time at least $3 + \pi/6 + \sqrt{3}/2 \approx 4.3896$ in the worst case. If one or both of the queen's other servants (Biddy and/or Lili) are with her, we show
that the time bounds
can be improved to $3.8327$ for two servants, and $3.3738$ for three servants. Finally we show lower bounds for these cases of $3.6307$  (two servants) and $3.2017$  (three servants). 
%The case of $n\geq4$ servants is discussed in a separate paper. 
The case of $n\geq 4$ is the subject of an independent study by Queen Daniela's Royal Scientific Team~\cite{asymtotic18}. 
\end{abstract}

\section{Introduction}

In traditional search, a group of searchers (modeled as mobile autonomous agents or robots) may collaboratively search for an exit (or target) placed within a given search domain~\cite{ahlswede1987search,alpern2002theory,stone1975theory}. Although the searchers may have differing capabilities (communication, perception, mobility, memory) search algorithms, previously employed, generally make no distinction between them as they usually play identical roles throughout the execution of the search algorithm and with respect to the termination time (with the exception of faulty robots, which also do not contribute to searching). 
In this work we are motivated by real-life safeguarding-type situations where a number of agents have the exclusive role to facilitate the execution of the task by a distinguished entity. More particularly, we introduce and study \textit{Priority Evacuation}, a new form of search , under the wireless communication model, in which the search time of the algorithm is measured by the time it takes a special searcher, called the queen, to reach the exit. The remaining searchers in the group, called servants, are participating in the search but are not required to exit. 
%Instead their goal is to help the queen and as such improve the time it takes her to locate the exit. 

%\textbf{Problem Definition of \textit{Priority Evacuation} (\pe{n}):}
\subsection{Problem Definition of \textit{Priority Evacuation} (\pe{n})}
A target (exit) is hidden in an unknown location on the unit circle. 
The exit can be located by any of the $n+1$ robots (searchers) that walks over it ($n=1,2,3$). Robots share the same coordinate system, start from the center of the circle, and have maximum speed 1. Among them there is a distinguished robot, called the \textit{queen}, and the remaining $n$ robots are referred to as \textit{servants}. 
All servants are known to the queen by their identities. Robots may run asymmetric algorithms, and can communicate their findings wirelessly and instantaneously (each message is composed by an identity and a location). Only the queen is required to be able to receive messages. 
Feasible solutions to this problem are \textit{evacuation algorithms}, i.e. robots' movements (trajectories) that guarantee the finding of the hidden exit. The cost of an evacuation algorithm is the \textit{evacuation time} of the queen, i.e., the worst case total time until the queen reaches the exit. None of the $n$ servants needs to evacuate.

\subsection{Related work}

Related to our work is linear search which refers to search in an infinite line. There have been several interesting studies attempting to optimize the search time which were initiated with the influential works of Bellman~\cite{bellman1963optimal} and Beck~\cite{beck1964linear}. % In these studies the authors are interested to minimize the competitive ratio in a stochastic setting. %, thus proving that time $9d$ is necessary and sufficient to ensure that an exit (or target) situated at an unknown) distance $d$ from the origin is found by the searcher. 
A long list of results followed for numerous variants of the problem, citing which is outside the scope of this work. For a comprehensive study of seminal search-type problems see~\cite{alpern2002theory, Alpern2013}. 

%Moreover, several other works on linear search followed (e.g. see \cite{alpern2002theory,beck1964linear,beck1965more,beck1984linear,beck1970yet,beck1973return,bellman1963optimal}) and more recently the search by a single searcher was studied for different models, e.g., when the turn cost was considered \cite{demaine2006online}, when the bounds on the distance to the target are known in advance \cite{Bose13}, and when the target was moving or for more general linear cost functions \cite{Bose16}.

%ADD PAPERS ON 2D EVACUATION
%To the best of our knowledge, the evacuation problem proposed in this paper is new and has never been investigated in the past. In the sequel we mention related research on search and evacuation in the unit disk model, linear search, as well as search in the plane.

The problem of searching in the 
%two-dimensional 
plane by one or more searchers, has been considered by \cite{baezayates1993searching,BS95}. The unit disk model considered in our present paper is a form of two-dimensional search that was initiated in the work of~\cite{CGGKMP}. In this paper the authors obtained evacuation algorithms in the wireless and face-to-face communication models both for a small number of robots as well optimal asymptotic results for a large number of robots. Additional evacuation algorithms in the face-to-face communication model were subsequently analyzed for two robots in \cite{DBLP:conf/ciac/CzyzowiczGKNOV15} and later in \cite{Watten2017}. 
Other variations of the problem include the case of more than one exit, see \cite{DBLP:conf/icdcn/CzyzowiczDGKM16} and \cite{pattanayak2017evacuating}, triangular and square domains in \cite{DBLP:conf/adhoc-now/CzyzowiczKKNOS15}, robots with different moving speeds \cite{lamprou2016fast}, and evacuation in the presence of crash or byzantine faulty robots \cite{georgioudiskfaulty2017}.

Notably, all relevant previous work in search-type problems considered the objective of minimizing the time it takes either by the first or the last agent to reach the hidden target. In contrast, this paper considers an evacuation (search-type) problem where the completion time is defined with respect to a distinguished mobile agent, the \textit{queen}, while the remaining $n$ servants are not required to evacuate. 
%, which is part of the search robot team (we refer to the remaining robots as \textit{servants}). We call our problem \pe{n}, to emphasize that the role of the $n$ servants is to prioritize the evacuation of the queen. 
Our current focus is to design efficient algorithms for $n=1,2,3$ servants, as well as give strong lower bounds. Notably, the algorithms we propose significantly improve upon evacuation costs induced by naive trajectories, and in fact the trajectories we propose are non-trivial. Our main contribution concerns priority evacuation for each of the cases of $n=1,2,3$ servants, all of which require special treatment. Moreover, all our algorithms are characterized by the fact that the queen does contribute effectively to the search of the hidden item. In sharp contrast, the independent and concurrent work of~\cite{asymtotic18} studies the same problem for $n\geq4$ servants where the queen never contributes to the search. More importantly, the proposed algorithms of~\cite{asymtotic18} admit a unified description and analysis that does not intersect with the current work.

\subsection{Our Results \& Paper Organization}

Section~\ref{sec2} introduces necessary notation and terminology and discusses preliminaries.  Section~\ref{ssec3} is devoted to upper bounds for \pe{n} for $n=1,2,3$ servants (see Subsections~\ref{sec3},~\ref{sec4},~and~\ref{sec5}, respectively). 
All our upper bounds are achieved by fixing optimal parameters for families of parameterized algorithms. 
In Section~\ref{seclbounds} we derive lower bounds for \pe{n}, $n=1,2,3$.
An interesting corollary of our positive results is that priority evacuation with $n=1,2,3$ servants (i.e. with $n+1$ searchers) can be performed strictly faster than ordinary evacuation with $n+1$ robots where all robots have to evacuate. Indeed, an argument found in \cite{CGGKMP} can be adjusted to show that the evacuation problem with $n+1$ robots cannot be solved faster than $ 1+\frac{4\pi}{3(n+1)} + \sqrt{3}$.
 %(even though the bounds are not believed to be tight). 
 Surprisingly, when one needs to evacuate only one designated robot, the task can provably (due to our upper bounds) be executed faster. 
All our results, together with the comparison to the lower bounds of \cite{CGGKMP}, are summarized in Table~\ref{defaultresults}. 
We conclude the paper in Section~\ref{secconclusion} with a discussion of open problems. 
%Upper and lower bounds as well as asymptoyic results for $n\geq 4$ servants can be found in the forthcoming \cite{asymtotic18}. 
%All proofs missing from the main text can be found in the Appendix.
\begin{table}[htp]
\caption{Upper bounds (UB) and lower bounds (LB) for priority evacuation.}
\begin{center}
\begin{tabular}{| c | l|l|l|}
\hline
\# of Servants & UB for \pe{n} & LB for \pe{n} 
& LB for Ordinary Evacuation 
\\
\hline
$n=1$ & $4.8185$ (Theorem~\ref{thm: 1 servant}) & $4.3896$  (Theorem~\ref{thm:lb2})
& 4.826445 (see \cite{CGGKMP})
\\
\hline
$n=2$ & $3.8327$ (Theorem~\ref{thm: 2 servants})& $3.6307$ (Theorem~\ref{thm:lb_n2})
& 4.128314 (see \cite{CGGKMP})\\
\hline
$n=3$ & $3.3738$ (Theorem~\ref{thm: 3 servants better})& $3.2017$ (Theorem~\ref{thm:lb_n3})
& 3.779248 (see \cite{CGGKMP})\\
 \hline
\end{tabular}
\end{center}
\label{defaultresults}
\end{table}%

\section{Notation and Preliminaries}
\label{sec2}

We use $n$ to denote the number of servants, and we set $[n]=\{1,\ldots,n\}$. Queen and servant $i$ will be denoted by $\queen$ and 
$\ser{i}$, respectively, where $i\in [n]$. 
We assume that all robots start from the origin 
$O=(0,0)$ of a unit circle in $\reals^2$. 
As usual, points in $A \in \reals^2$ will be treated, when it is convenient, as vectors from $O$ to $A$, and $\norm{A}$ will denote the euclidean norm of that vector.

\subsection{Problem Reformulation \& Solutions' Description}

Robots' trajectories will be defined by parametric functions $\mathcal{F}(t)=(f(t), g(t))$, where $f,g:\reals\mapsto \reals$ are continuous and piecewise differentiable. 
In particular, search algorithms for all robots will be given by trajectories 
$$
\sol{n}:=\left\{
\queen(t)
, \{\ser{i}(t)\}_{i\in [n]}
\right\},$$
where $\queen(t), \ser{i}(t)$ will denote the position of $\queen$ and $\ser{i}$, respectively, at time $t\geq 0$. 

\begin{definition}[Feasible Trajectories]
We say that trajectories 
$\sol{n}$
are \textit{feasible} for \pe{n} if: 
\begin{enumerate}[(a)]
\item $\queen(0)=\ser{i}(0)=O$, for all $i \in [n]$, 
\item $\queen(t), \{\ser{i}(t)\}_{i\in [n]}$ induce speed-1 trajectories for $\queen, \{\ser{i}\}_{i\in [n]}$ respectively, and
\item there is some time $t_0\geq 1$, such that each point of the unit circle is visited (searched) by at least one robot in the time window $[0,t_0]$. We refer to the smallest such $t_0$ as the \textit{search time} of the circle. 
\end{enumerate}
\end{definition}

Note that feasible trajectories do indeed correspond to robots' movements for \pe{n} in which, eventually the entire circle is searched, and hence the search time is bounded. We will describe all our search/evacuation algorithms as feasible trajectories, and we will assume that once the target is reported, $\queen$ will go directly to the location of the exit. 

For feasible trajectories $\sol{n}$ with search time $t_0$, and for any trajectory $\mathcal{F}(t)$ (either of the queen or of a servant), we denote by $\interval{\mathcal{F}}$ the subinterval of $[0,t_0]$ that contains all $x \in [0,t_0]$ such that $\norm{\mathcal{F}(x)}=1$ (i.e. the robot is on the the circle) and no other robot has been to $\mathcal{F}(x)$ before. Since robots start from the origin, it is immediate that $\interval{\mathcal{F}} \subseteq [1,t_0]$. With this notation in mind, note that the exit can be discovered by some robot $\mathcal{F}$, say at time $x$, only if $x \in \interval{\mathcal{F}}$. In this case, the finding is instantaneously reported, so $\queen$ goes directly to the exit, moving along the corresponding line segment between her current position $\queen(x)$ and the reported position of the exit $\mathcal{F}(x)$. Hence, the total time that $\queen$ needs to evacuate equals 
$$
x+\norm{\queen(x)-\mathcal{F}(x)}. 
$$
Therefore, the \textit{evacuation time} of feasible trajectories $\sol{n}$ to \pe{n} is given by expression 
$$
\max_{\mathcal{F} \in \sol{n}}
\sup_{x \in \interval{\mathcal{F}}}
\left\{
x+\norm{\queen(x)-\mathcal{F}(x)}
\right\}
.$$
Notice that for ``non-degenerate'' search algorithms for which the last point on the circle is not searched by $\queen$ alone, the previous maximum can be simply computed over the servants, i.e the evacuation cost will be 
\begin{equation}\label{equa: total cost}
\max_{i\in [n]}
\sup_{x \in \interval{\ser{i}}}
\left\{
x+\norm{\queen(x)-\ser{i}(x)}
\right\}.
\end{equation}
In other words, we can restate \pe{n}\ as the problem of determining feasible trajectories $\sol{n}$ so as to minimize~\eqref{equa: total cost}.

\subsection{Useful Trajectories' Components}

Feasible trajectories induce, by definition, robots that are moving at (maximum) speed 1. 
The speed restriction will be ensured by the next condition. 

\begin{lemma}\label{lem: unit speed}
An object following trajectory $\mathcal{F}(t)=(f(t), g(t))$ has unit speed if and only if
$$
\left(f'(t)\right)^2 + \left(g'(t)\right)^2 =1, ~~\forall t\geq0.
$$
\end{lemma}

\begin{proof}%[Proof of Lemma~\ref{lem: unit speed}]
For any $t\geq 0$, the velocity of $\mathcal F$ is given by $\mathcal{F}'(t)=(d f(t)/dt, d g(t)/dt)$, and its speed is calculated as $\norm{\mathcal{F}'(t)}$. 
\ignore{
the total distance traversed by the object between time 0 and $t$ equals 
$$
D(t):=\int_0^t  \sqrt{\left( \frac{d f}{dx} \right)^2 +  \left(\frac{d g}{dx} \right)^2 } dx. 
$$
Clearly, the speed of the object is 1 iff for each $t$ we have $t=D(t)$. By the fundamental theorem of calculus, this implies that $1=\sqrt{\left(f'(t)\right)^2 + \left(g'(t)\right)^2}$. Finally, if the latter is true, we conclude that $D(t)=t+c$ for some constant $c$. However $D(0)=0$, hence $c=0$.  
}
\end{proof}

Robots' trajectories will be composed by piecewise smooth parametric functions. In order to describe them, we introduce some further notation. For any $\theta\in \reals$, we introduce abbreviation $\ki{\theta}$ for point $\{\coss{\theta}, \sinn{\theta}\}$. Next we introduce parametric equations for moving along the perimeter of a unit circle (Lemma~\ref{lem: move circle}), and along a line segment (Lemma~\ref{lem: move line}). 

\begin{lemma}\label{lem: move circle}
Let $b \in [0,2\pi)$ and $\sigma \in \{-1,1\}$. The trajectory of an object moving at speed 1 on the perimeter of a unit circle with initial location $\ki{b}$ is given by the parametric equation 
$$ 
\ci{b,\sigma t}:=(\coss{\sigma t+b}, \sinn{\sigma t+b}).
$$
%\ki{\sigma t+b}.$$ 
If $\sigma=1$ the movement is counter-clockwise (ccw), and clockwise (cw) otherwise. 
\end{lemma}
\begin{proof}%[Proof of Lemma~\ref{lem: move circle}]
Clearly, $\ci{b,0}=\ki{b}$. Also, it is easy to see that $\norm{\ci{b,t}}=1$, i.e. the object is moving on the perimeter of the unit circle. Lastly, 
$$
\left(\frac{d}{dt} \coss{\sigma t+b}\right)^2 + 
\left(\frac{d}{dt} \sinn{\sigma t+b}\right)^2 
=
\sigma^2 \left(- \sinn{\sigma t+b}\right)^2 + 
\sigma^2 \left( \coss{\sigma t+b}\right)^2 =1, 
$$
so the claim follows by Lemma~\ref{lem: unit speed}. 
\end{proof}

\begin{lemma}\label{lem: move line}
Consider distinct points $A=(a_1,a_2), B=(b_1,b_2)$ in $\reals^2$. The trajectory of a speed 1 object moving along the line passing through $A,B$ and with initial position $A$ is given by the parametric equation 
$$
\li{A,B,t}:=\left(
\frac{b_1-a_1}{\norm{A-B}}t+a_1,
\frac{b_2-a_2}{\norm{A-B}}t+a_2
\right).
$$
\end{lemma}

\begin{proof}%[Proof of Lemma~\ref{lem: move line}]
It is immediate that the parametric equation corresponds to a line. Also, it is easy to see that $\li{A,B,0}=A$ and $\li{A,B,\norm{A-B}}=B$, i.e. the object starts from $A$, and eventually visits $B$. As for the object's speed, we calculate
$$
\left(\frac{d}{dt} \left(\frac{b_1-a_1}{\norm{A-B}}t+a_1\right)\right)^2 + 
\left(\frac{d}{dt} \left(
\frac{b_2-a_2}{\norm{A-B}}t+a_2
\right)
\right)^2 
=
\left(
\frac{b_1-a_1}{\norm{A-B}}
\right)^2 + 
\left(
\frac{b_2-a_2}{\norm{A-B}}
\right)^2 =1
$$
so, by Lemma~\ref{lem: unit speed}, the speed is indeed 1. 
\end{proof}

Robots trajectories will be described in phases. 
In each phase, robot, say $\mathcal F$, will be moving between two explicit points, and the corresponding trajectory $\mathcal{F}(t)$ will be implied by the previous description, using most of the times Lemma~\ref{lem: move circle} and Lemma~\ref{lem: move line}.
We will summarize the details in tables of the following format. 
\begin{center}
\begin{tabular}{l || lllc}
\textit{Robot}		 	&	\# & \textit{Description} & \textit{Trajectory} & \textit{Duration}\\
\hline 
$\mathcal F$	&	0			&		&	$\mathcal{F}(t)$ & $t_0$ \\
				&	1			&		&	$\mathcal{F}(t)$ &  $t_1$ \\
				&	$\vdots$		&		&	&  $\vdots$	\\
\hline			
\end{tabular}
\end{center}

Phase 0 will usually correspond to the deployment of $\mathcal F$ from the origin to some point of the circle. 
Also, for each phase we will summarize it's duration. With that in mind, trajectory $\mathcal{F}(t)$ during phase $i$, with duration $t_i$, will be valid for all $t\geq 0$ with $|t-(t_0+t_1+\ldots t_{i-1})|\leq t_i$.

Lastly, the following abbreviation will be useful for the exposition of the trajectories. For any $\rho \in [0,1]$ and $\theta \in [0,2\pi)$, we introduce notation 
$$
K(\theta, \rho):= (1-\rho) \ki{\pi - \theta} + \rho \ki{-\theta}.
$$
In other words, $K(\theta, \rho)$ is a convex combination of antipodal points $\ki{\pi - \theta}, \ki{- \theta}$ of the unit circle, i.e. it lies on the diameter of the unit circle passing through these two points. Moreover, it is easy to see that $\norm{\ki{\pi-\theta}- 
K(\theta, \rho)
}=2\rho
$, and hence
$$\norm{K(\theta, \rho)-\ki{-\theta}}=2-2\rho.$$ 
As it will be handy later, we also introduce abbreviation 
$$
AK(\theta,\rho):=\norm{\ki{\pi}- K(\theta,\rho)}.
$$
The choice of the abbreviation is clear, if the reader denotes $\ki{\pi}=(-1,0)$ by $A$.

\subsection{Critical Angles}

The following definition introduces a key concept. In what follows, abstract trajectories will be assumed to be continuous and differentiable, which in particular implies that corresponding velocities are continuous. 

\begin{figure}[h!]
%\begin{wrapfigure}{h!}{0.5\textwidth}
  \centering
  \includegraphics[width=.5\linewidth]{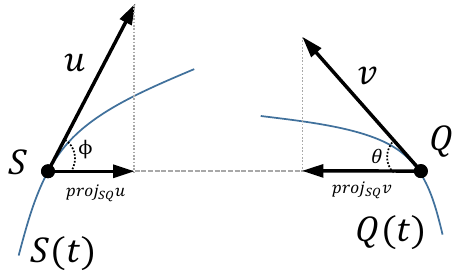}
\caption{An illustration of trajectories $\ser{}(t), \queen(t)$, and their critical angles at some fixed time $\tau$, with $\ser{}(\tau)=S, \queen(\tau)=Q,\ser{}'(\tau)=u, \queen'(\tau)=v$.}
\label{fig:CriticalAngles}
%\end{wrapfigure}
\end{figure}

\begin{definition}[Critical Angle]
Let $\ser{}(t) \in \reals^2$ denote the trajectory of a speed-1 object, where $t\geq 0$. For some point $Q \in \reals^2$, we define the $(\ser{},Q)$-critical angle at time $t=\tau$ to be the angle between the velocity vector $\ser{}'(\tau)$ and vector 
$\overrightarrow{\ser{}(\tau)Q}$, i.e. the vector from $\ser{}(\tau)$ to $Q$. 
\end{definition}

We make the following critical observation, see also Figure~\ref{fig:CriticalAngles}.
\begin{theorem}\label{thm: critical angles and cost}
Consider trajectories $\ser{}(t), \queen(t)$ of two speed-1 objects $\ser{}, \queen$, where $t\geq 0$.
Let also $\phi, \theta$ denote 
the $(\ser{},\queen(t))$-critical angle
and the $(\queen, \ser{}(t))$-critical angle at time $t$, respectively. 
Then $t+\norm{\queen(t) - \ser{}(t)}$ is strictly increasing if $\coss{\phi}+\coss{\theta}< 1$, strictly decreasing if 
$\coss{\phi}+\coss{\theta}> 1$, and constant otherwise. 
\end{theorem}

%An immediate corollary of Lemma~\ref{thm: cost preservation} is that for trajectories $\ser{}(t), \queen(t)$ as in the lemma, we have that $t+\norm{\queen(t) - \ser{}(t)}$ is strictly increasing if $\coss{\phi}+\coss{\theta}< 1$, strictly decreasing if $\coss{\phi}+\coss{\theta}> 1$, and constant otherwise. 

Theorem~\ref{thm: critical angles and cost} is an immediate corollary of the following lemma. 
\begin{lemma}
\label{thm: cost preservation}
Consider trajectories $\ser{}(t), \queen(t)$ and their critical angles $\pi, \theta$, as in the statement of Theorem~\ref{thm: critical angles and cost}. Then 
$$
\frac{d}{dt} \norm{\queen(t) - \ser{}(t)} = \coss{\phi}+\coss{\theta}. 
$$
\end{lemma}
\begin{proof}%[Proof of Theorem~\ref{thm: critical angles and cost}]
For any fixed $t$, let $d$ denote $D(t)$, and $S,Q$ denote points $\ser{}(t), \queen(t)$, respectively. Denote 
also by $u,v$ the velocities of $\ser{}, \queen$ at time $t$, respectively, i.e. $u=\ser{}'(t), v=\queen'(t)$. See also Figure~\ref{fig:CriticalAngles}.

With that notation, observe that $\norm{\overrightarrow{SQ}}=d$. Since $\norm{u}=\norm{v}=1$, we see that 
%\begin{align*}
$$
\textrm{proj}_{SQ} u = \frac{\coss{\phi}}d \overrightarrow{SQ}
$$
and 
$$
\textrm{proj}_{SQ} v = \frac{\coss{\theta}}d \overrightarrow{QS}.
$$
%\end{align*}
Now consider two imaginary objects $\overline{\ser{}}, \overline{\queen}$, with corresponding velocities $\overline{\ser{}}'(t)=\textrm{proj}_{SQ} u$ and $\overline{\queen}'(t)=\textrm{proj}_{SQ} v$. 
It is immediate that $\norm{\queen(t) - \ser{}(t)}= \norm{\overline{\queen}(t) - \overline{\ser{}}(t)}$. 

In particular, $\textrm{proj}_{SQ} u-\textrm{proj}_{SQ} v$ is the projection of the relative velocities of $\ser{}, \queen$ on the line segment connecting $\ser{}(t), \queen(t)$. 
As such, the distance between $\ser{}, \queen$ changes at a rate determined by velocity
$$
\textrm{proj}_{SQ} u-\textrm{proj}_{SQ} v
=\frac{\coss{\phi}+\coss{\theta}}d \overrightarrow{SQ},
$$
where $\norm{\textrm{proj}_{SQ} u-\textrm{proj}_{SQ} v} = \left| \coss{\phi}+\coss{\theta} \right|$. 
Moreover, $\textrm{proj}_{SQ} u,\textrm{proj}_{SQ} v$ are antiparallel iff and only if $\coss{\phi}, \coss{\theta}> 0$, in which case the two objects come closer to each other. 
\end{proof}

\section{Upper Bounds}
%\section{Saving the Queen With $n=1$ Servant}
\label{ssec3}

%In this section we give upper bounds for the cases of $n=1,2,3$ servants. 

\subsection{Evacuation Algorithm for \pe{1}}
\label{sec3}

This subsection is devoted in proving the following. 

\begin{theorem}\label{thm: 1 servant}
Consider the real function $f(x)=x+\sinn{x}$, and denote by $\alpha_0>0$ the solution to equation
$$
f(f(\alpha-\sinn{\alpha}))=\sinn{\alpha},
$$
with $\alpha_0 \approx 1.14193$. Then 
\pe{1} can be solved in time $1+ \pi-\alpha_0 +2\sinn{\alpha_0} \approx 4.81854$. 
\end{theorem}

The value of $\alpha_0$ is well defined in the statement of Theorem~\ref{thm: 1 servant}. Indeed, by letting $g(x)=f(f(x-\sinn{x}))-\sinn{x}$, we observe that $g$ is continuous, while $g(1)\approx -0.213934$ and $g(\pi/2)\approx 1.00729$, hence there exists $\alpha_0\in (1,\pi/2)$ with $g(\alpha_0)=0$. 

In order to prove Theorem~\ref{thm: 1 servant}, 
and given parameters $\alpha, \beta$, we introduce the family of trajectories $\textsc{Search}_1(\alpha,\beta)$, see also Figure~\ref{fig:Search1NEW-S1}.
\begin{figure}[h!]
\centering
  \includegraphics[width=.61\linewidth]{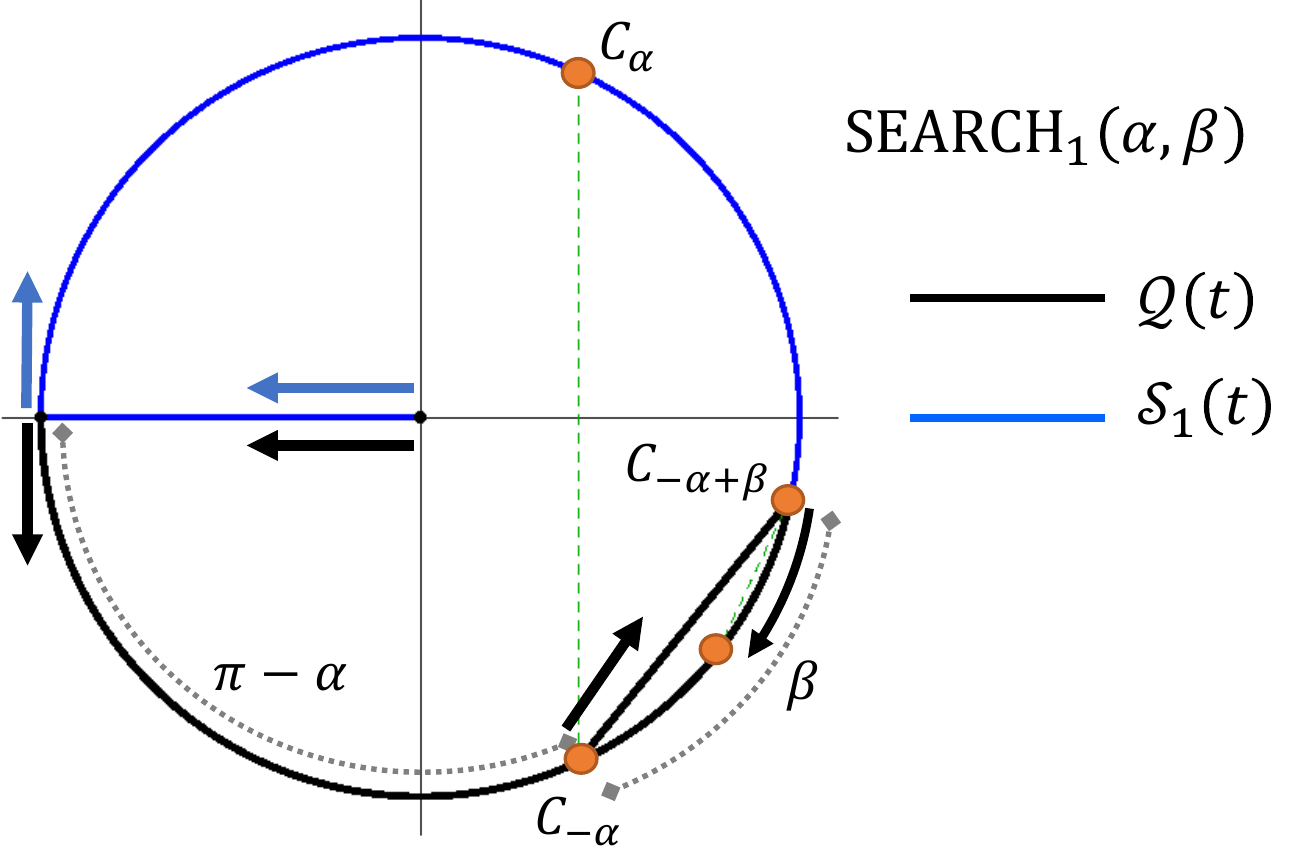}
\caption{
Algorithm $\textsc{Search}_1(\alpha,\beta)$ depicted for the optimal parameters of the algorithm. 
%Robot's trajectories are also clarified by some velocity vectors (matching the aforementioned colours) at special points.
In all subsequent figures, as well as here, the orange points on the perimeter of the disc correspond to the worst adversarial placements of the treasure, which due to our optimality conditions induce the same evacuation cost.
The orange points in $\queen$'s trajectories correspond to the $\queen$'s positioning when the treasures are reported, in  the worst cost induced cases.
The green dashed line depict $\queen$'s trajectory after $\queen$ abandons her trajectory and moves toward the reported exit following a straight line. }
\label{fig:Search1NEW-S1}
\end{figure}

\small{
%\vspace{-0.5cm}
\begin{center}
\begin{tabular}{|l|}
\hline
\textbf{Algorithm $\textsc{Search}_1(\alpha,\beta)$} \\
\hline \hline
\begin{tabular}{l || llll}
\textit{Robot}		 	&	\# & \textit{Description} & \textit{Trajectory} & \textit{Duration}\\
\hline 
$\queen$	&	0		&	Move to point $\ki{\pi}$	&	$\li{O,\ki{\pi}, t}$										&	$1$					\\
			&	1		&	Search circle ccw till point $\ki{-\alpha}$			&$\ci{\pi,t-1}$												&	$\pi-\alpha$	\\
			& 2		&Move to point $\ki{-\alpha+\beta}$, 						&$\li{\ki{-\alpha}, \ki{-\alpha+\beta},t-(1+\pi-\alpha)}$ &$2\sinn{\beta/2}$ \\
			& 3	&	Search circle cw till point $\ki{-\alpha}$			&$\ci{\beta-\alpha,1+\pi-\alpha+2\sinn{\beta/2}-t}$												&	$\beta$	\\
\hline 
$\ser{1}$	&	0		&	Move to point $\ki{\pi}$							&	$\li{O,\ki{\pi}, t}$										&	$1$					\\
			&	1		&	Search circle cw till point $\ki{\beta-\alpha}$			&$\ci{\pi,-t+1}$												&	$\pi+\alpha-\beta$					\\
\hline 
\end{tabular}
\end{tabular}
\end{center}
}

Partitioning  the circle clockwise, we see that the arc with endpoints $\ki{\pi }, \ki{\pi+\alpha-\beta}$ is searched by $\ser{1}$, while the remaining of the circle is searched by $\queen$. Therefore, robots' trajectories in $\textsc{Search}_1(\alpha,\beta)$ are feasible, and it is also easy to see that they are continuous as well. 
The search time equals $1+\pi + \max\{\alpha-\beta, 2\sinn{\beta/2}+\beta-\alpha\}$, as well as 
$$
\interval{\queen}=[1,1+\pi - \alpha] \cup [1+\pi-\alpha+2\sinn{\beta/2}, 1+\pi-\alpha+2\sinn{\beta/2}+\beta], 
\interval{\ser{1}}=[1, 1+\pi+\alpha-\beta]
.$$
An illustration of the above trajectories for certain values of $\alpha,\beta$ can be seen in Figure~\ref{fig:Search1NEW-S1}.
\ignore{
\begin{figure}[h!]
\centering
  \includegraphics[width=.31\linewidth]{newSearch1-2D.pdf}
~~~ \includegraphics[width=.31\linewidth]{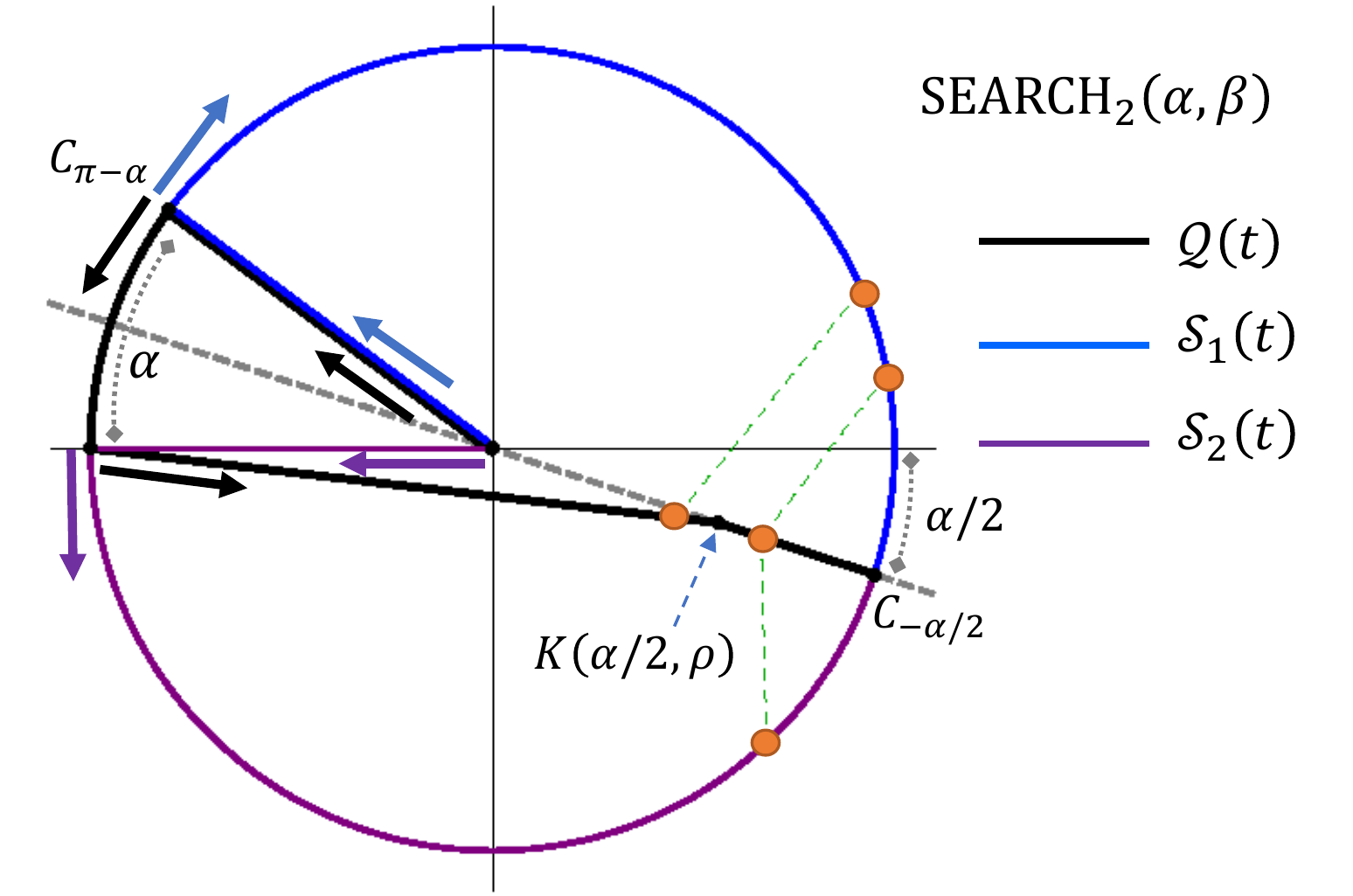}  
~~~~  \includegraphics[width=.33\linewidth]{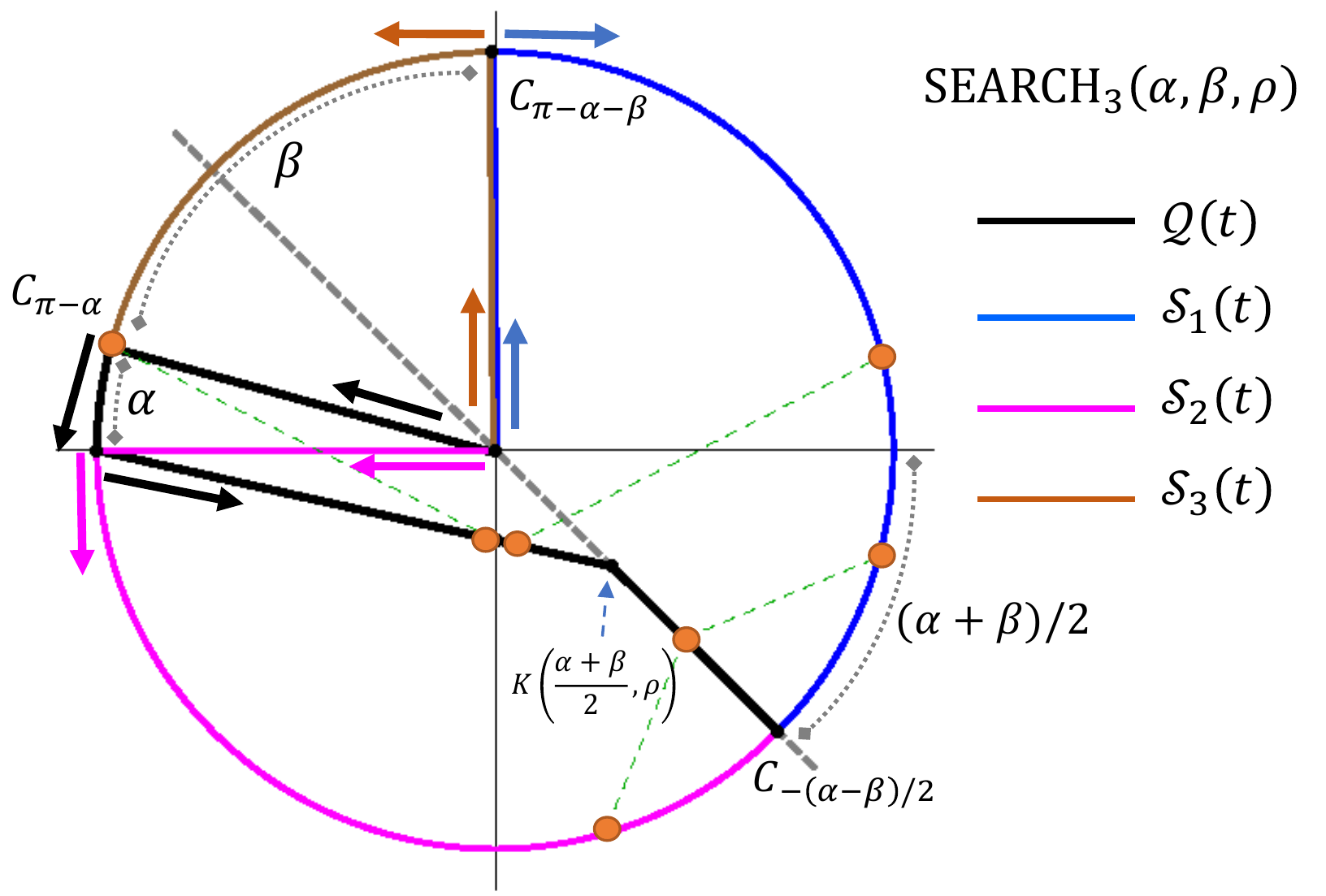}
~~~ \includegraphics[width=.33\linewidth]{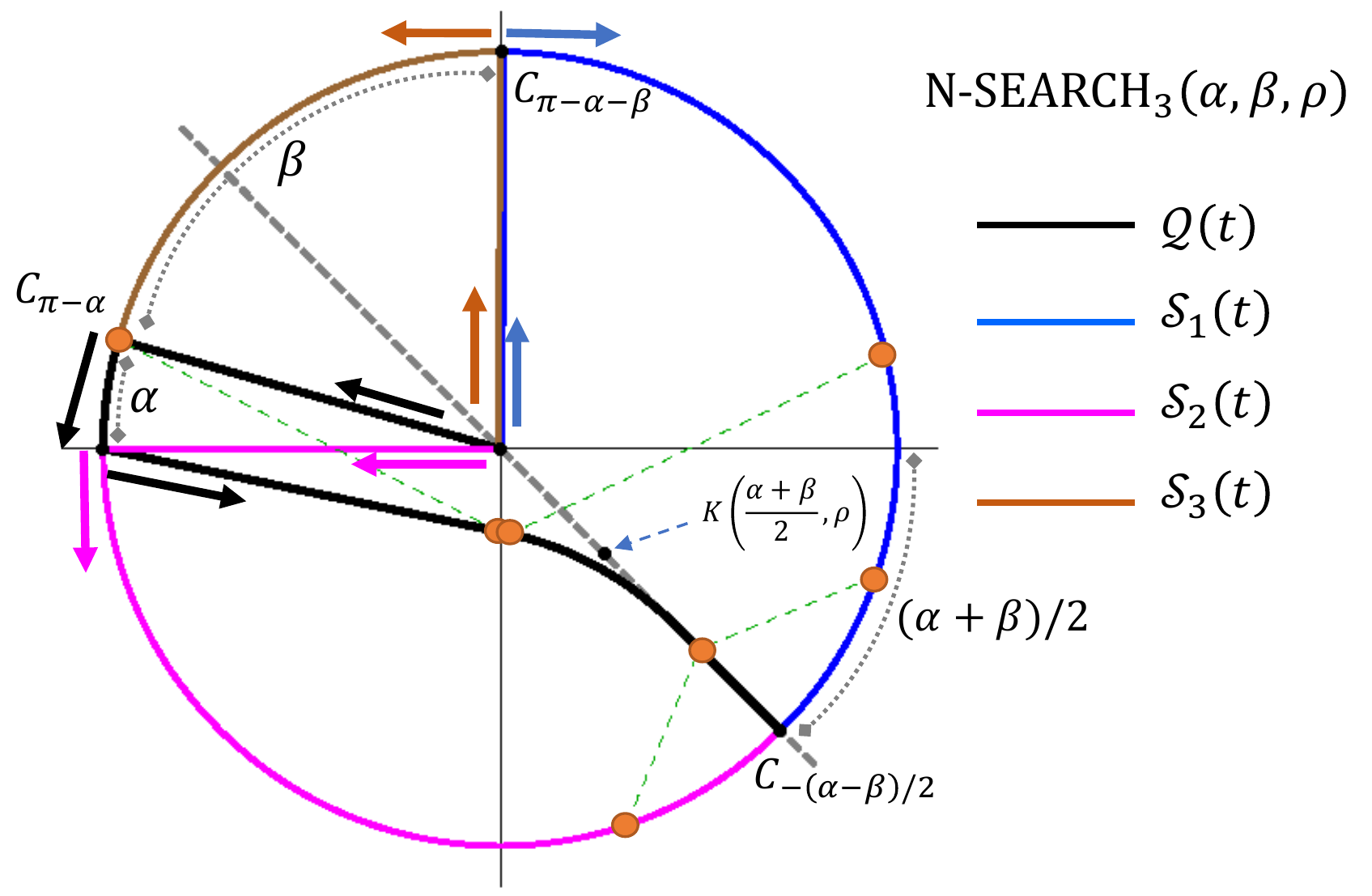}  
\caption{
All trajectories are depicted for the optimal parameters of our algorithms. 
%Robot's trajectories are also clarified by some velocity vectors (matching the aforementioned colours) at special points.
The orange points on the perimeter of the disc correspond to the worst adversarial placements of the treasure, which due to our optimality conditions induce the same evacuation cost.
The orange points in $\queen$'s trajectories correspond to the $\queen$'s positioning when the treasures are reported, in  the worst cost induced cases.
The green dashed line depict $\queen$'s trajectory after $\queen$ abandons her trajectory and moves toward the reported exit following a straight line. }
\label{fig:Search1NEW}
\end{figure}
}

First we make some observations pertaining to the monotonicity of the evacuation cost. 

\begin{lemma}\label{lem: monotonicity search 1}
Assuming that $\alpha>\pi/3$ and that $\coss{\alpha} + \coss{\alpha-\beta/2} >1$, the evacuation cost of 
$\textsc{Search}_1(\alpha,\beta)$ is monotonically increasing if the exit is found by $\ser{1}$ during $\queen$'s phase 1 and monotonically decreasing if the exit is found by $\ser{1}$ during $\queen$'s phase 2.
\end{lemma}

\begin{proof}%[Proof of Lemma~\ref{lem: monotonicity search 1}]
Suppose that the exit is found by $\ser{1}$ during $\queen$'s phase 1, i.e. at time $x$ after robots start searching for the first time, where $0\leq x \leq \pi - \alpha$. It is easy to see that the critical angles between $\queen, \ser{1}$ are both equal to $\pi-x$. 
But then $2\coss{\pi-x} \geq 2\coss{\alpha} > 2\coss{\pi/3} =1$. Hence, by Theorem~\ref{thm: critical angles and cost}, the evacuation cost is decreasing in this case. 

Now suppose that the exit is found by $\ser{1}$ during $\queen$'s phase 2, i.e. at time $x$ after $\queen$ starts moving along the chord with endpoints $\ki{-\alpha}, \ki{-\alpha+\beta}$, where $0\leq x \leq 2\sinn{\beta/2}$. 
If $\phi_x, \theta_x$ denote the $\ser{1}, \queen$ critical angles, then it is easy to see that 
$\phi_0 = \coss{\alpha}$ and that $\theta_0 = \alpha-\beta/2$. 
Since $\coss{\phi_0} + \coss{\theta_0} >1$, Theorem~\ref{thm: critical angles and cost} implies that the evacuation cost is initially  decreasing in this phase. For the remaining of $\queen$'s phase 2, it is easy to see that both $\phi_x, \theta_x$ are decreasing in $x$, hence $\coss{\phi_x} + \coss{\theta_x}$ is increasing in $x$, hence, the evacuation cost will remain decreasing in this phase. 
\end{proof}

Now we can prove Theorem~\ref{thm: 1 servant} by fixing certain values for parameters $\alpha, \beta$ of $\textsc{Search}_1(\alpha,\beta)$. In particular, we set $\alpha_0$ as in the statement of Theorem~\ref{thm: 1 servant}, and $\beta_0 = 2f(\alpha_0-\sinn{\alpha_0}) \approx 0.925793$. The trajectories of the robots, for the exact same values of the parameters, can be seen in Figure~\ref{fig:Search1NEW-S1}. 

\begin{proof}[Proof of Theorem~\ref{thm: 1 servant}] \\
Let $f,\alpha_0$ be as in the statement of Theorem, and set $\beta_0 = 2f(\alpha_0-\sinn{\alpha_0}) \approx 0.925793$. We argue that the worst evacuation time of $\textsc{Search}_1(\alpha_0,\beta_0)$ is $1+\pi-\alpha_0+2\sinn{\alpha_0}$. 
Note that for the given values of the parameters, we have that $\alpha_0>\pi/3$, that 
$\alpha_0-\sinn{\beta_0/2}\leq \beta_0$, and that $\coss{\alpha_0} + \coss{\alpha_0-\beta_0/2} >1$. 

First we observe that if the exit if found by $\queen$, then the worst case evacuation time $E_0(\alpha_0, \beta_0) $ is incurred when the exit is found just before $\queen$ stops searching, that is 
$$E_0(\alpha_0, \beta_0) =1+\pi-\alpha_0+2\sinn{\beta_0/2}+\beta_0.$$

Next we examine some cases as to when the exit is found by $\ser{1}$. 
If the exit is found by $\ser{1}$ during the 1st phase of $\queen$, then the evacuation time is, due to Lemma~\ref{lem: monotonicity search 1}, given as 
$$
E_1(\alpha_0,\beta_0) = \sup_{1\leq x \leq 1+\pi-\alpha_0}
\left\{
x+\norm{\queen(x)-\ser{1}(x)}
\right\} = 1+\pi-\alpha_0+2\sinn{\alpha_0}.
$$

Recall that $\coss{\alpha_0} + \coss{\alpha_0-\beta_0/2} >1$, and so, again by Lemma~\ref{lem: monotonicity search 1} we may omit the case that the exit is found by $\ser{1}$ while $\queen$ is at phase 2. The end of $\queen$'s phase 2 happens at time $\tau:=1+\pi-\alpha_0+2\sinn{\beta_0/2}$, when have that $\queen(\tau)=\ki{-\alpha+\beta}$, and $\ser{1}(\tau)=\ki{\alpha - 2 \sinn{\beta_0/2}}$, and both robots are intending to search ccw. 
Condition $\alpha_0-\sinn{\beta_0/2}\leq \beta_0$ says that $\ser{1}$ will finish searching prior to $\queen$, and this happens when $\ser{1}$ reaches point $\ki{-\alpha+\beta}$. During this phase, the distance between $\queen, \ser{1}$ stays invariant and equal to $2\alpha_0-\beta_0-2\sinn{\beta_0/2}$. We conclude that the cost in this case would be 
$$
E_2(\alpha_0,\beta_0) = 
 1+\pi+\alpha_0-\beta_0 + 2\sinn{\alpha_0-\beta_0/2-\sinn{\beta_0/2}}.
$$
Then, we argue that that the choice of $\alpha_0, \beta_0$ guarantees that $E_0(\alpha_0,\beta_0)=E_1(\alpha_0,\beta_0)=E_2(\alpha_0,\beta_0)$, as wanted.

Indeed, $E_0(\alpha_0,\beta_0)=E_1(\alpha_0,\beta_0)$ implies that 
$\sinn{\beta_0/2}+\beta_0/2 =\sinn{\alpha_0}$. But then, we can rewrite $E_2(\alpha_0,\beta_0)$ as
$$
E_2(\alpha_0,\beta_0) = 
 1+\pi+\alpha_0-\beta_0 + 2\sinn{\alpha_0-\sinn{\alpha_0}}.
$$
Equating the last expression with $E_1(\alpha_0,\beta_0)$ implies that 
$$
\beta_0/2=\alpha_0-\sinn{\alpha_0}+\sinn{\alpha_0-\sinn{\alpha_0}} = f(\alpha_0-\sinn{\alpha_0}). 
$$
Substituting twice $\beta_0/2$ in the already derived condition $\sinn{\beta_0/2}+\beta_0/2 =\sinn{\alpha_0}$ implies that $$f(f(\alpha-\sinn{\alpha_0}))=\sinn{\alpha_0}.$$ 
Figure~\ref{fig:Search1NEW-S1} depicts the worst placements of the exit, along with the trajectories of the queen (in dashed green lines) after the exit is reported. 
\end{proof}

%It should be stressed that Phases 2 and 3 of $\queen$ are essential in achieving a bound that is better than that of the evacuation algorithm for two robots of~\cite{CGGKMP} (which is $\textsc{Search}_1(0,0)$, inducing cost $1+2\pi/3+\sqrt{3} \approx 4.82645$). 

%\ignore{
It should be stressed that $\queen$'s Phases 2,3 are essential for achieving the promised bound. Indeed, had we chosen $\alpha=\beta=0$, the worst case evacuation time would have been
$$
\sup_{1\leq x \leq 1+\pi}
\left\{
x+\norm{\queen(x)-\ser{1}(x)}
\right\} 
=
\sup_{0 \leq x \leq \pi}
\left\{
1+x+2\sinn{x}
\right\}.
$$
The maximum is attained at $x_0=2\pi/3$ (and indeed, both critical angles in this case are $\pi/3$ and in particular $2\coss{\pi/3}=1$), inducing cost $1+2\pi/3+\sqrt{3} \approx 4.82645$. The latter is the cost of the evacuation algorithm for two robots without priority of~\cite{CGGKMP}.
%}

\subsection{Evacuation Algorithm for \pe{2}}
\label{sec4}

%\marginpar{We should prove a lower bound for $n=2$ servants.}
In this subsection we prove the following theorem. 

\begin{theorem}
\label{thm: 2 servants}
\pe{2} can be solved in time 3.8327. 
\end{theorem}

Given parameters $\alpha, \rho$, we introduce the family of trajectories $\textsc{Search}_2(\alpha,\rho)$, see also Figure~\ref{fig:Search1NEW-S2}.
\begin{figure}[h!]
\centering
  \includegraphics[width=.61\linewidth]{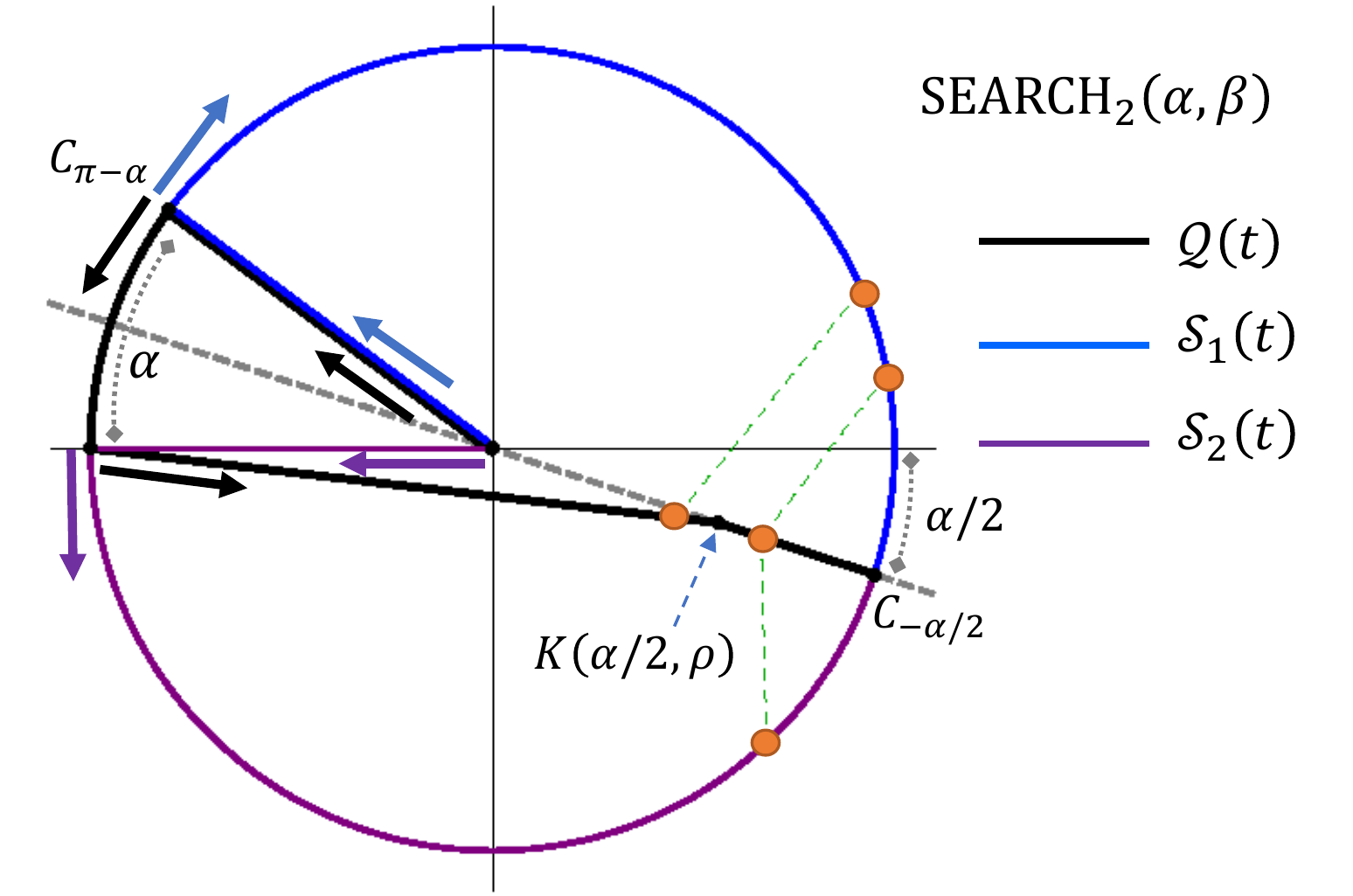}
\caption{
Algorithm $\textsc{Search}_2(\alpha,\beta)$ depicted for the optimal parameters of the algorithm. 
%Robot's trajectories are also clarified by some velocity vectors (matching the aforementioned colours) at special points.
}
\label{fig:Search1NEW-S2}
\end{figure}

%, corresponding to robots $\queen, \ser{1}, \ser{2}$.%, whose evacuation time is denoted by $E_2$.\\
%All robots start from the origin O, and they follow the specified trajectories till the exit is reported, in which case $\queen$ abandons her trajectory and goes directly to the exit (moving along the corresponding line segment between her current position and the reported position of the exit). 
%\vspace{-0.5cm}
\small{
\begin{center}
\begin{tabular}{|l|}
\hline
\textbf{Algorithm $\textsc{Search}_2(\alpha,\rho)$} \\
\hline \hline
\begin{tabular}{l || lllc}
\textit{Robot}		 	&	\# & \textit{Description} & \textit{Trajectory} & \textit{Duration}\\
\hline 
$\queen$	&	0		&	Move to point $\ki{\pi-\alpha}$										&$\li{O,\ki{\pi-\alpha}, t}$										&	$1$					\\
			&	1		&	Search the circle ccw till point $\ki{\pi}$			&$\ci{\pi-\alpha,t-1}$												&	$\alpha$					\\
			&	2		&	Move to point $K(\alpha/2, \rho)$										&$\li{\ki{\pi}, K(\alpha/2,\rho),t-(1+\alpha)}$				&	$AK(\alpha/2,\rho)$					\\
			&	3		&	Move to point $\ki{-\alpha/2}$										&$\li{K(\alpha/2,\rho),\ki{-\alpha/2}}$							&	$2-2\rho$					\\
\hline 
$\ser{1}$	&	0		&	Move to point $\ki{\pi-\alpha}$								&	$\li{O,\ki{\pi-\alpha}}$			&	1					\\
			&	1			& Search the circle cw till point $\ki{-\alpha/2}$				&	$\ci{\pi-\alpha,-t+1}$			&	$\pi-\alpha/2$					\\
\hline 
$\ser{2}$	&	0		& Move to point $\ki{\pi}$										&	$\li{O,\ki{\pi}}$			&	1					\\
			&	1			& Search the circle cw till point $\ki{-\alpha/2}$				&	$\ci{\pi,t-1}$			&	$\pi-\alpha/2$					\\
\hline			
\end{tabular}
\end{tabular}
\end{center}
}

Notice that, by definition of $\textsc{Search}_2(\alpha,\rho)$, robots' trajectories are continuous and feasible, meaning that the entire circle is eventually searched. Indeed, partitioning  the circle clockwise, we see that: 
the arc with endpoints $\ki{\pi }, \ki{\pi-\alpha}$ is searched by $\queen$, 
the arc with endpoints $\ki{\pi-\alpha }, \ki{-\alpha/2}$ is searched by $\ser{1}$, 
and the arc with endpoints $\ki{-\alpha/2 }, \ki{\pi}$ is searched by $\ser{2}$. 

It is immediate from the description of the trajectories that the search time is $1+\pi-\alpha/2$. Moreover 
$$
\interval{\queen}=[1,1+\alpha],~ 
\interval{\ser{1}}=
\interval{\ser{2}}=[1, 1+\pi-\alpha/2]
.$$
An illustration of the above trajectories for certain values of $\alpha,\rho$ can be seen in 
Figure~\ref{fig:Search1NEW-S2}.
%Figure~\ref{fig:Search2}.
%\begin{figure}[h!]
%\centering
%  \includegraphics[width=.3\linewidth]{Search2-2D.pdf}
%\caption{The trajectories of $\queen$ in black, $\ser{1}$ in blue and $\ser{2}$ in magenta. }
%\label{fig:Search2}
%\end{figure}
Now we make some observations, in order to calculate the worst case evacuation time. %as described in Formula~\eqref{equa: total cost}.

\begin{lemma}
\label{lem: cost analysis for Search2}
Suppose that $\pi-\alpha/2 \geq \alpha+AK(\alpha/2,\rho)+2-2\rho$. 
Then $\norm{\queen(x)-\ser{1}(t)}$ is continuous and differentiable in the time intervals 
$I_1,I_2,I_3$ of $\queen$'s phases 1,2,3, respectively. 
Moreover, the worst case evacuation time of $\textsc{Search}_2(\alpha,\rho)$ can be computed as 
$$
\max
\left\{
\begin{array}{l}
1+\alpha+2\sinn{\alpha}, \\
\sup_{t \in I_2 }
\left\{ t+\norm{\queen(t)-\ser{1}(t)} \right\} \\
\sup_{t \in I_3 }
\left\{ t+\norm{\queen(t)-\ser{1}(t)} \right\} \\
1+\pi-\alpha/2
\end{array}
\right\}
$$
where
$$ I_2=[1+\alpha, 1+\alpha+AK(\alpha/2,\rho)], I_3=[1+\alpha+AK(\alpha/2,\rho),3-2\rho+\alpha+AK(\alpha/2,\rho) ].$$
\end{lemma}

\begin{proof}%[Proof of Lemma~\ref{lem: cost analysis for Search2}]
Note that the line passing through $O$ and $\ki{-\alpha/2}$, call it $\epsilon$, has the property that each point of it, including $K(\alpha/2, \rho)$ is equidistant from $\ser{1}, \ser{2}$.
Moreover, in the time window $[1+\alpha, 1+\alpha+AK(\alpha/2, \rho)]$ that only $\ser{1},\ser{2}$ are searching, $\queen$ stays below line $\epsilon$.
At time $1+\alpha+AK(\alpha/2, \rho)$, $\queen$ is, by construction, equidistant from $\ser{1}, \ser{2}$, a property that is preserved for the remaining of the execution of the algorithm. 
As a result, the evacuation time of $\textsc{Search}_2(\alpha,\rho)$ is given by $\sup_{1\leq t \leq 1+\pi-\alpha/2 }\{
t+\norm{\queen(t)-\ser{1}(t)}
\}$.

Now note that condition $\pi-\alpha/2 \geq \alpha+AK(\alpha/2,\rho)+2-2\rho$ guarantees that $\queen$ reaches point $\ki{-\alpha/2}$ no later than $\ser{1}$. Moreover, in each time interval $I_1,I_2,I_3$, $\queen$'s trajectory is differentiable (and so is $\ser{1}$'s trajectory). 
\end{proof}

Now Theorem~\ref{thm: 2 servants} can be proven by fixing parameters $\alpha,\rho$ for $\textsc{Search}_2(\alpha,\rho)$, in particular, $\alpha=0.6361, \rho=0.7944$. Notably, the performance of $\textsc{Search}_2(\alpha,\rho)$ is provably improvable (slightly) using a technique we will describe in the next section. 

\begin{proof}[Proof of Theorem~\ref{thm: 2 servants}]
We choose $\alpha=0.6361, \rho=0.7944$. The trajectories of Figure~\ref{fig:Search1NEW-S2} correspond exactly to those values. 
The time that $\queen$ needs to reach $\ki{-\alpha/2}$ equals $1+\alpha+AK(\alpha/2,\rho)+2-2\rho=3.6174$, while the time that $\ser{1}, \ser{2}$ reach the same point is $1+\pi-\alpha/2 =3.82354$. Hence, Lemma~\ref{lem: cost analysis for Search2} applies.

The worst case evacuation time during phase 1 is $1+a+2\sinn{\alpha}=2.82423$. 
The worst case evacuation time after $\queen$ reaches $\ki{-\alpha/2}$, equals $1+\pi-\alpha/2 =3.82354$.
Hence, it remains to compute the maxima of $t+\norm{\queen(t)-\ser{1}(t)}$ in the two intervals $I_2, I_3$, which can be done numerically using the trajectories of $\textsc{Search}_2(\alpha,\rho)$, since expression is differentiable in each of the intervals. 

To that end, when $t \in I_2=[1.6361, 3.2062]$ we have that 
\begin{align*}
\queen(t) &=\left(0.9931 t-2.62481,0.191866\, -0.11727 t \right) \\
\ser{1}(t) &=\left(\coss{3.50549\, -t},\sinn{3.50549\, -t}\right), 
\end{align*}
so that $t+\norm{\queen(t)-\ser{1}(t)}$ becomes
$$
t+
\sqrt{(-\sinn{3.50549\, -t}-0.11727 t+0.191866)^2+(-\coss{3.50549\, -t)+0.9931 t-2.62481}^2}
$$
When $t \in I_3=[3.2062, 3.6174]$ we have that
\begin{align*}
\queen(t) &= \left(0.949847 t-2.48613,0.818501\, -0.312715 t\right) 
\end{align*}
while $\ser{1}$'s trajectory equation remains unchanged, 
so that $t+\norm{\queen(t)-\ser{1}(t)}$ becomes
$$
t+
\sqrt{(-\sinn{3.50549\, -t)-0.312715 t+0.818501}^2+(-\coss{3.50549\, -t)+0.949847 t-2.48613}^2}
$$

In particular, it follows that 
$$
\sup_{t \in I_2 }
\left\{ t+\norm{\queen(t)-\ser{1}(t)} \right\}
\approx 
\sup_{t \in I_3 }
\left\{ t+\norm{\queen(t)-\ser{1}(t)} \right\}
\approx 3.8327\\
$$
with corresponding maximizers (with approximate values) $\tau_2=3.10066$ and $\tau_3=3.32114$, respectively.
Figure~\ref{fig:Search1NEW-S2} also depicts the locations of the optimizers, i.e the worst case locations on the circle for the exit to be found, along with the corresponding evacuation trajectory in dashed green colour.
\end{proof}

%\newpage

% 3 servants------------------------------------------------
% 3 servants------------------------------------------------
% 3 servants------------------------------------------------
% 3 servants------------------------------------------------
% 3 servants------------------------------------------------

\subsection{Evacuation Algorithm for \pe{3}}
\label{sec5}

%\marginpar{We should prove a lower bound for $n=3$ servants.}

%In this section we analyze the case of three servants.

\subsubsection{A Simple Algorithm}
\label{sec: 3 servant simple}

In this section we prove the following preliminary theorem (to be improved in Section~\ref{sec: 3 servant cost pres}).

\begin{theorem}
\label{thm: 3 servants}
\pe{3} can be solved in time 3.37882. 
%\todo{put here best known value}
\end{theorem}

Given parameters $\alpha,\beta, \rho$, we introduce the family of trajectories $\textsc{Search}_3(\alpha,\beta,\rho)$, corresponding to robots $\queen, \ser{1}, \ser{2},\ser{3}$, see also Figure~\ref{fig:Search1NEW-S3}.
\begin{figure}[h!]
\centering
  \includegraphics[width=.61\linewidth]{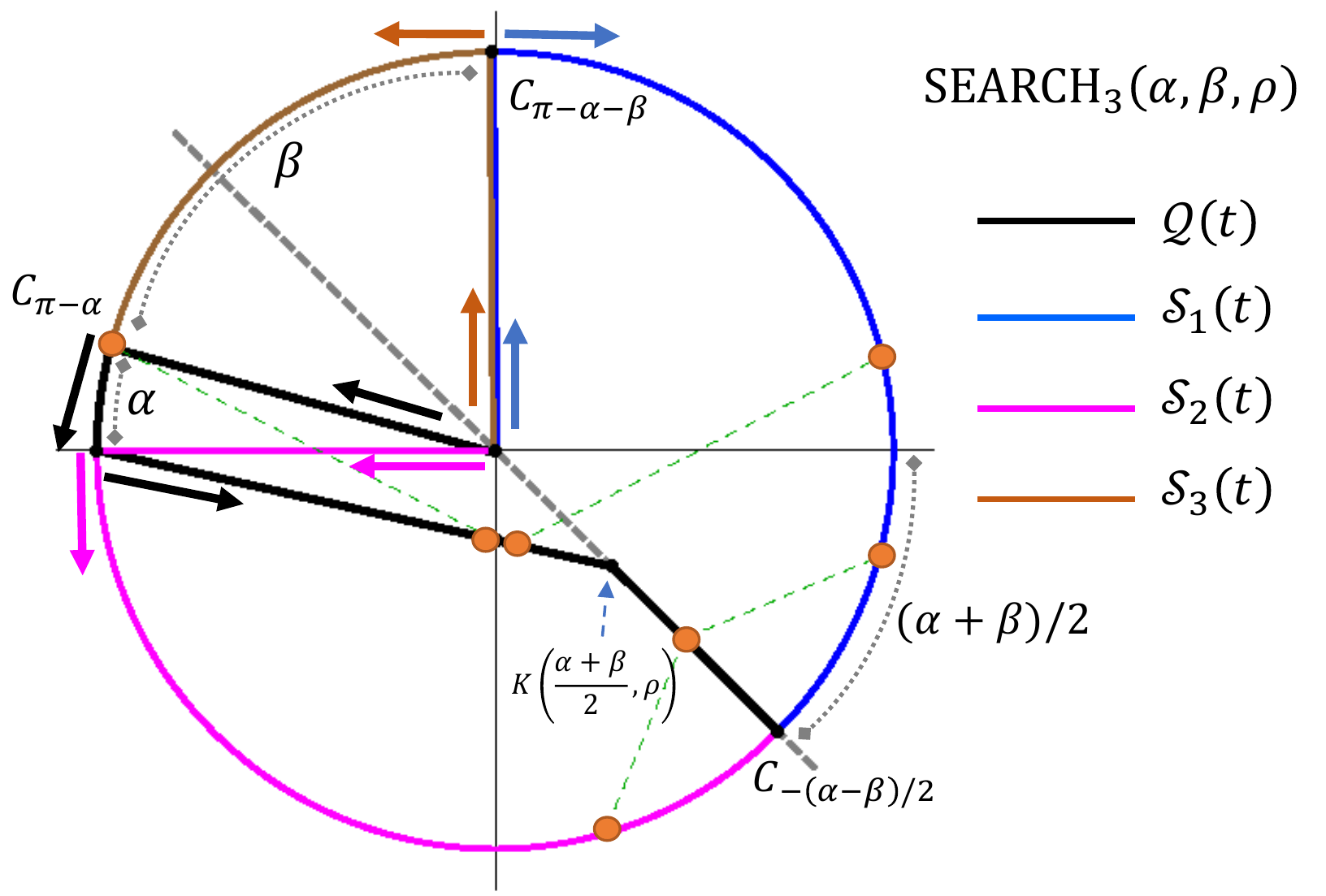}
\caption{
Algorithm $\textsc{Search}_3(\alpha,\beta,\rho)$ depicted for the optimal parameters of the algorithm. 
%Robot's trajectories are also clarified by some velocity vectors (matching the aforementioned colours) at special points.
}
\label{fig:Search1NEW-S3}
\end{figure}
%All robots start from the origin O, and they follow the specified trajectories till the exit is reported, in which case $\queen$ abandons her trajectory and goes directly to the exit (moving along the corresponding line segment between her current position and the reported position of the exit). 
%\vspace{-0.5cm}
\small{
\begin{center}
\begin{tabular}{|l|}
\hline
\textbf{Algorithm $\textsc{Search}_3(\alpha,\beta, \rho)$} \\
\hline \hline
\begin{tabular}{l || lllc}
\textit{Robot}		 	&	\# & \textit{Description} & \textit{Trajectory} & \textit{Duration}\\
\hline 
$\queen$	&	0		&	Move to point $\ki{\pi-\alpha}$										&$\li{O,\ki{\pi-\alpha}, t}$										&	$1$					\\
			&	1		&	Search the circle ccw till point $\ki{\pi}$			&$\ci{\pi-\alpha,t-1}$												&	$\alpha$					\\
			&	2		&	Move to point $K(\frac{\alpha+\beta}2, \rho)$										&$\li{\ki{\pi}, K(\frac{\alpha+\beta}2,\rho),t-(1+\alpha)}$				&	$AK(\frac{\alpha+\beta}2,\rho)$					\\
			&	3		&	Move to point $\ki{-\frac{\alpha+\beta}2}$										&$\li{K(\frac{\alpha+\beta}2,\rho),\ki{-\frac{\alpha+\beta}2}}$							&	$2-2\rho$					\\
\hline 
$\ser{1}$	&	0		&	Move to point $\ki{\pi-\alpha-\beta}$								&	$\li{O,\ki{\pi-\alpha-\beta}}$			&	1					\\
			&	1			& Search the circle cw till point $\ki{-\frac{\alpha+\beta}2}$				&	$\ci{\pi-\alpha-\beta,-t+1}$			&	$\pi-\frac{\alpha+\beta}2$					\\
\hline 
$\ser{2}$	&	0		&	Move to point $\ki{\pi}$								&	$\li{O,\ki{\pi}}$			&	1					\\
			&	1			& Search the circle ccw till point $\ki{-\frac{\alpha+\beta}2}$				&	$\ci{\pi,t-1}$			&	$\pi-\frac{\alpha+\beta}2$					\\
\hline 
$\ser{3}$	&	0		&	Move to point $\ki{\pi-\alpha-\beta}$								&	$\li{O,\ki{\pi-\alpha-\beta}}$			&	1					\\
			&	1			& Search the circle ccw till point $\ki{-\alpha}$				&	$\ci{\pi-\alpha-\beta,-t+1}$			&	$\beta$					\\
\hline 

\end{tabular}
\end{tabular}
\end{center}
}

As before, it is immediate that, in $\textsc{Search}_3(\alpha,\beta,\rho)$, robots' trajectories are continuous and feasible, meaning that the entire circle is eventually searched.
In particular, 
the arc with endpoints $\ki{\pi }, \ki{\pi-\alpha}$ is searched by $\queen$, 
the arc with endpoints $\ki{\pi-\alpha-\beta }, \ki{-\frac{\alpha+\beta}2}$ is searched by $\ser{1}$, 
the arc with endpoints $\ki{-\pi}, \ki{-\frac{\alpha+\beta}2}$ is searched by $\ser{2}$, 
and the arc with endpoints $\ki{\pi-\alpha }, \ki{\pi -\alpha-\beta}$ is searched by $\ser{3}$.
Also, the search time is $1+\pi-\frac{\alpha+\beta}2$, and
$$
\interval{\queen}=[1,1+\alpha],~ 
\interval{\ser{1}}=
\interval{\ser{2}}=[1, 1+\pi-\frac{\alpha+\beta}2],~
\interval{\ser{3}}=[1, 1+\beta]
.$$
An illustration of the above trajectories for certain values of $\alpha,\beta,\rho$ can be seen in Figure~\ref{fig:Search1NEW-S3}.

\ignore{
\begin{figure}[h!]
\centering
  \includegraphics[width=.33\linewidth]{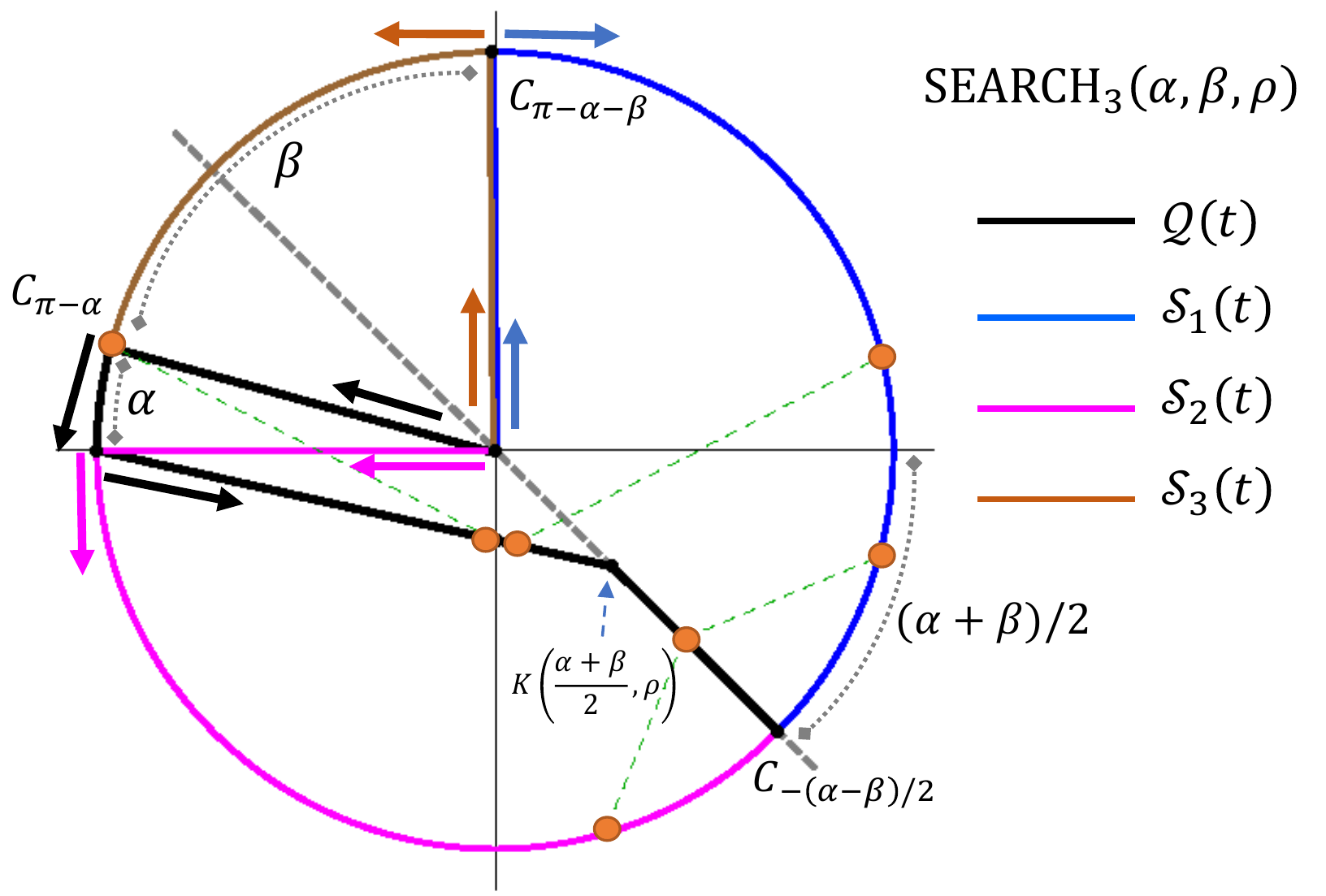}
~~~
  \includegraphics[width=.33\linewidth]{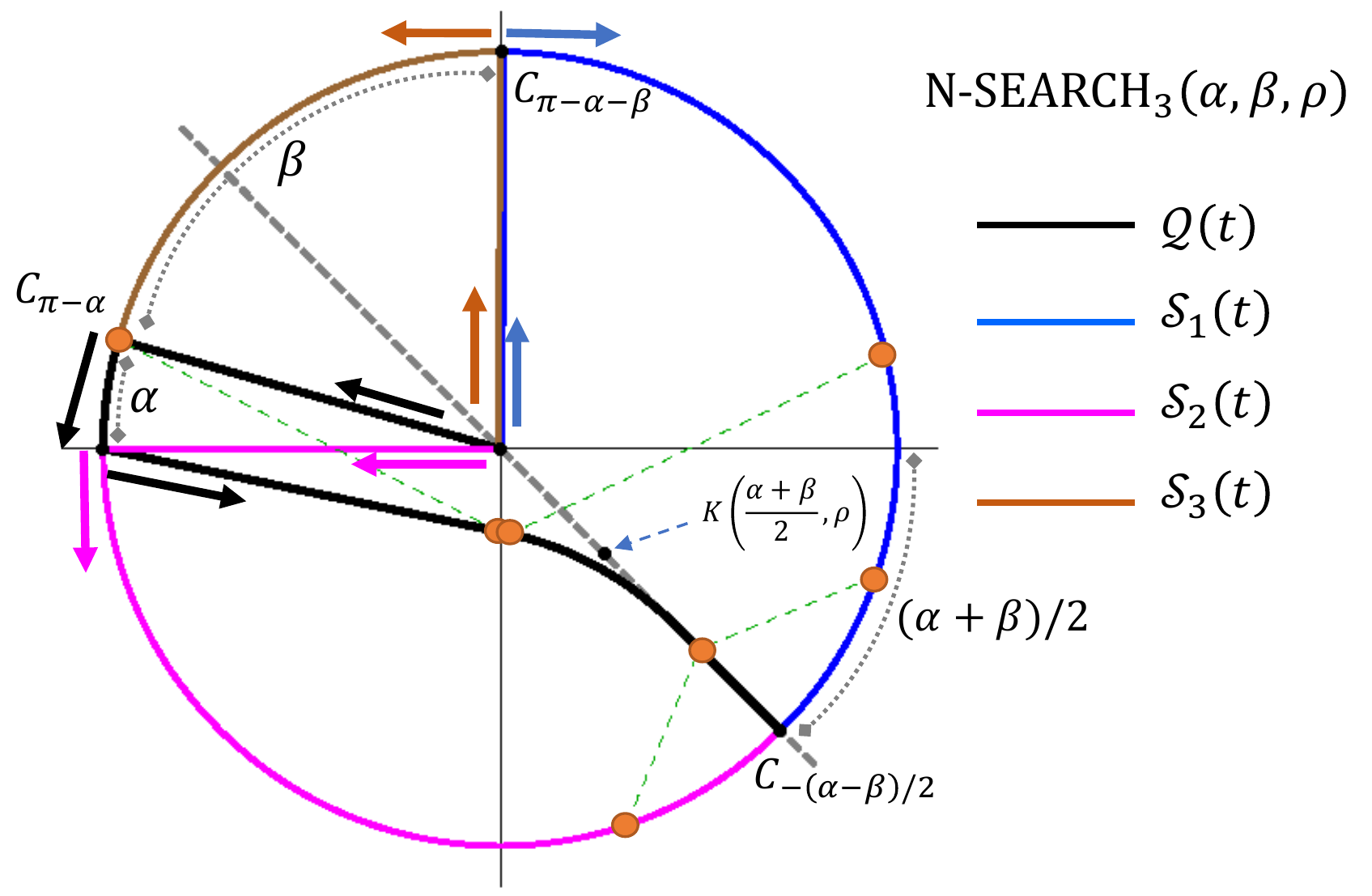}  
\caption{(Left: Simple Algorithm) The trajectories of $\queen$ in black, $\ser{1}$ in blue, $\ser{2}$
 in magenta and $\ser{3}$ in brown. (Right: Improved Algorithm) The trajectories of $\queen$ in black, $\ser{1}$ in blue, $\ser{2}$ in magenta and $\ser{3}$ in brown. 
Robot's trajectories are also clarified by some velocity vectors (matching the aforementioned colours) at special points.
The parameters of the algorithms are depicted in both figures.
Finally the orange points on the perimeter of the disc correspond to the worst adversarial placements of the treasure, which due to our optimality conditions induce the same evacuation cost.
The orange points in queen's trajectories correspond to the queen's positioning when the treasures are reported, in  the worst cost induced cases.
The green dashed line depict queens' trajectory after queen abandones her trajectory and moves toward the reported exit following a straight line.
}
\label{fig:Search3}
\end{figure}
}

Before we prove Theorem~\ref{thm: 3 servants}, we need to make some observation, in order to calculate the worst case evacuation time. %as described in Formula~\eqref{equa: total cost}.

\begin{lemma}
\label{lem: cost analysis for Search3}
Suppose that 
$
\alpha\leq \beta
$,
$\alpha+AK(\frac{\alpha+\beta}2,\rho)\geq \beta$,
and
$\pi-\frac{\alpha+\beta}2 \geq \alpha+AK(\frac{\alpha+\beta}2,\rho)+2-2\rho$.
Then the following functions are continuous and differentiable in each associated time intervals:
$\norm{\queen(x)-\ser{3}(t)}$ 
in $I_1=\{t\geq 0:~\alpha\leq t-1\leq \beta \}$,
$\norm{\queen(x)-\ser{1}(t)}$ 
in $I_2=\{t \geq 0:~ |t-1-\alpha| \leq AK(\frac{\alpha+\beta}2,\rho)\}$ and in $I_3=\{t\geq 0:~|t-1-\alpha-AK(\frac{\alpha+\beta}2,\rho)| \leq 2-2\rho\}$.
Moreover, the worst case evacuation time of $\textsc{Search}_3(\alpha,\beta,\rho)$ can be computed as 
$$
\max
\left\{
\begin{array}{l}
\sup_{t \in I_1 }
\left\{ t+\norm{\queen(t)-\ser{3}(t)} \right\} \\
\sup_{t \in I_2 }
\left\{ t+\norm{\queen(t)-\ser{1}(t)} \right\} \\
\sup_{t \in I_3 }
\left\{ t+\norm{\queen(t)-\ser{1}(t)} \right\} \\
1+\pi-\frac{\alpha+\beta}2
\end{array}
\right\}
$$
\end{lemma}

\begin{proof}%[Proof of Lemma~\ref{lem: cost analysis for Search3}]
Conditions 
$
\alpha\leq \beta
$ and
$\alpha+AK(\frac{\alpha+\beta}2,\rho)\geq \beta$
mean that
$\queen$ stops searching no later than $\ser{3}$, 
and that when $\ser{3}$ stops searching $\queen$ is still in her phase 2, 
respectively. 

The line passing through $O$ and $\ki{-(\alpha+\beta)/2}$, call it $\epsilon$, has the property that each point of it, including $K(\frac{\alpha+\beta}2, \rho)$ is equidistant from $\ser{1}, \ser{2}$.
Moreover, while $\ser{1},\ser{2}$ are searching, $\queen$ never goes above line $\epsilon$.
At time $1+\alpha+AK(\frac{\alpha+\beta}2, \rho)$, $\queen$ is, by construction, equidistant from $\ser{1}, \ser{2}$, a property that is preserved for the remaining of the execution of the algorithm. 
As a result, $\ser{2}$ can be ignored in the performance analysis, and when it comes to the case that $\ser{1}$ finds the exit, the evacuation cost is given by the supremum of  $t+\norm{\queen(t)-\ser{1}(t)}$ in the time interval $I_2$ or in the interval $I_3$. Note that in both intervals, the evacuation cost is continuous and differentiable, by construction. 

If the exit is reported by $\ser{3}$ then the evacuation cost is $t+\norm{\queen(t)-\ser{3}(t)}$ for $t\in [1,1+\beta]$. However, it is easy to see that the cost is strictly increasing for all $t\in [1,1+\alpha]$ (in fact it is linear). Since the evacuation cost is also continuous, we may restrict the analysis in interval $I_1$. 

Lastly, observe that $\pi-\frac{\alpha+\beta}2 \geq \alpha+AK(\frac{\alpha+\beta}2,\rho)+2-2\rho$
implies that $\ser{1},\ser{2}$ reach point $\ki{-(\alpha+\beta)/2}$ no earlier than $\queen$. Hence $\queen$ waits at $\ki{-(\alpha+\beta)/2}$ till the search of the circle is over, which can be easily seen to induce the worse evacuation time after $\queen$ reaches $\ki{-(\alpha+\beta)/2}$.
\end{proof}

Next, we prove Theorem~\ref{thm: 3 servants} by fixing parameters $\alpha,\beta,\rho$ for $\textsc{Search}_3(\alpha,\beta,\rho)$.
%, in particular $\alpha=0.26738, \beta=1.2949, \rho=0.70685$.

\begin{proof}[Proof of Theorem~\ref{thm: 3 servants}]
We choose 
$\alpha=0.26738, 
\beta=1.2949, 
\rho=0.70685$.
The trajectories of Figure~\ref{fig:Search1NEW-S3} correspond exactly to those values. 
The time that 
$\queen$ needs to reach $\ki{-\frac{\alpha+\beta}2}$ 
equals $1+\alpha+AK(\frac{\alpha+\beta}2,\rho)+2-2\rho=3.17984$, 
%1 + aaa + AK[aaa, bbb, rrr] + 2 - rrr
while the time that $\ser{1}, \ser{2}$ reach the same point is 
$1+\pi-\frac{\alpha+\beta}2 =3.36045$. 
%1 + Pi - (aaa + bbb)/2
Hence, Lemma~\ref{lem: cost analysis for Search3} applies.

From the above, it is immediate that the worst evacuation time after $\queen$ reaches $\ki{-(\alpha+\beta)/2}$ equals $1+\pi-\frac{\alpha+\beta}2 =3.36045$.
Hence, it remains to compute the maxima of $t+\norm{\queen(t)-\ser{3}(t)}$ in interval $I_1$, and of $t+\norm{\queen(t)-\ser{1}(t)}$ in intervals $I_2$, $I_3$. 

To that end, when $t \in I_1=[1.26738,2.2949]$ we have that 
%1 + aaa
%1 + bbb
\begin{align*}
\queen(t) &=\left(-2.23643 + 0.97558 t, 0.278372 - 0.219643 t \right) \\
%Expand[linePoints[pointA, pointK[aaa, bbb, rrr], t - aaa - 1]]
\ser{3}(t) &=\left(\cos (t+0.579313),\sin (t+0.579313)\right), 
%cycle[-aaa - bbb + t + Pi - 1]
\end{align*}
so that $t+\norm{\queen(t)-\ser{3}(t)}$ becomes
$$
t+
\sqrt{(-0.219643 t-\sin (t+0.579313)+0.278372)^2+(0.97558 t-\cos (t+0.579313)-2.23643)^2}
%Expand[newNorm[  Expand[linePoints[pointA, pointK[aaa, bbb, rrr], t - aaa - 1]] -    cycle[-aaa - bbb + t + Pi - 1]]]
$$
in which case
$$
\sup_{t \in I_1 }
\left\{ t+\norm{\queen(t)-\ser{3}(t)} \right\}
=
1+\beta+
\norm{\queen(1+\beta)-\ser{3}(1+\beta)}
\approx 3.37882
%costcase1[1 + bbb]
$$

When $t \in I_2=[1.26738,2.59354]$,
%1 + aaa
%1 + aaa + AK[aaa, bbb, rrr]
 $\queen$'s trajectory is the same as in $I_1$ and 
\begin{align*}
\ser{1}(t) &=\left(\cos (2.57931\, -t),\sin (2.57931\, -t)\right), 
%cycle[Pi - aaa - bbb - t + 1]
\end{align*}
so that $t+\norm{\queen(t)-\ser{1}(t)}$ becomes
$$
t+
\sqrt{(-\sin (2.57931\, -t)-0.219643 t+0.278372)^2+(-\cos (2.57931\, -t)+0.97558 t-2.23643)^2}.
%costcase2[t]
$$

When $t \in I_3=[2.59354,3.17984]$, $\ser{1}$'s trajectory is the same as in $I_2$ and
%1 + aaa + AK[aaa, bbb, rrr]
%1 + aaa + AK[aaa, bbb, rrr] + 2 - rrr
\begin{align*}
\queen(t) &=\left(-1.54793 + 0.710111 t, 1.5348 - 0.704089 t \right), 
%FullSimplify[ linePoints[pointK[aaa, bbb, rrr],   pointC[aaa, bbb], -1 + t - (aaa + AK[aaa, bbb, rrr])]]
\end{align*}
so that $t+\norm{\queen(t)-\ser{1}(t)}$ becomes
$$
t+
\sqrt{(\sin (2.57931\, -t)+0.704089 t-1.5348)^2+(\cos (2.57931\, -t)-0.710111 t+1.54793)^2}.
%t + newNorm[  cycle[Pi - aaa - bbb - t + 1] -    FullSimplify[    linePoints[pointK[aaa, bbb, rrr],      pointC[aaa, bbb], -1 + t - (aaa + AK[aaa, bbb, rrr])]]]
$$

Numerically 
\begin{align*}
\sup_{t \in I_2 }
\left\{ t+\norm{\queen(t)-\ser{1}(t)} \right\}
&=
 \tau_2+\norm{\queen(\tau_2)-\ser{1}(\tau_2)} 
\approx
3.37882 \\
%NMaximize[{costcase2[t], 1 + aaa <= t <= 1 + aaa + AK[aaa, bbb, rrr] }, {t}]
% {3.3789, {t -> 2.3404}}
\sup_{t \in I_3 }
\left\{ t+\norm{\queen(t)-\ser{1}(t)} \right\}
&=
 \tau_3+\norm{\queen(\tau_3)-\ser{1}(\tau_3)} 
\approx
3.37882
% NMaximize[{costcase3[t],  1 + aaa + AK[aaa, bbb, rrr] <= t <=    1 + aaa + AK[aaa, bbb, rrr] + 2 - rrr }, {t}]
%{3.37881, {t -> 2.84767}}
\end{align*}
where $\tau_2\approx2.34029$ and $\tau_3\approx2.84758$. 
% NMaximize[{costcase2[t], 1.26738 <= t <= 2.59354}, {t}]
% NMaximize[{costcase3[t], 2.59354 <= t <= 3.17984}, {t}]
%Figure~\ref{fig:Search3} depicts the locations of the optimizers, i.e the worst case locations on the circle for the exit to be found by any of the robots. Figure also shows corresponding locations of $\queen$ on her trajectory, along with the corresponding evacuation trajectory in dashed green colour.
\end{proof}

\subsubsection{Improved Search Algorithm}
\label{sec: 3 servant cost pres}
In this section we improve the upper bound of Theorem~\ref{thm: 3 servants} by 0.00495 additive term.

\begin{theorem}
\label{thm: 3 servants better}
\pe{3} can be solved in time 3.37387. 
\end{theorem}

The main idea can be described, at a high level, as a cost preservation technique. 
By the analysis of Algorithm $\textsc{Search}_3(\alpha,\beta, \rho)$ for the value of parameters of $\alpha,\beta,\rho$ as in the proof of Theorem~\ref{thm: 3 servants}, we know that there are is a critical time window $[\tau_2, \tau_3]$ so that the total evacuation time is the same if the exit is found by $\ser{1}$ either at time $\tau_2$ or $\tau_3$, and strictly less for time moments strictly in-between. In fact, during time $[\tau_2, 1+\alpha+AK(\frac{\alpha+\beta}{2},\rho)]$ $\queen$ is executing phase 2, and in the time window $[1+\alpha+AK(\frac{\alpha+\beta}{2},\rho), \tau_3]$ $\queen$ is executing phase 3 of $\textsc{Search}_3(\alpha,\beta, \rho)$. 

From the above, it is immediate that we can lower $\queen$'s speed in the time window $[\tau_2, \tau_3]$ so that the evacuation time remains \textit{unchanged} no matter when $\ser{1}$ finds the exit in the same time interval (notably, $\ser{3}$ has finished searching prior to $\tau_2$ and 
$\norm{\queen(t) - \ser{1}} \geq \norm{\queen(t) - \ser{2}}$). But this also implies that we must be able to maintain the evacuation time even if we preserve speed 1 for $\queen$, that will in turn allow us to twist parameters $\alpha,\beta, \rho$, hopefully improving the worst case evacuation time. 
We show this improvement is possible by using the
following technical observation

%\paragraph{A Technical Observation}
\begin{theorem}
\label{thm: cost preservation stationary point}
Consider point $Q=(q_1,q_2)\in \reals^2$. Let $\ser{}(t)$ be the trajectory of an object $\ser{}$ moving at speed 1, where $t\geq 0$, and denote by $\phi$ the $(\ser{},Q)$-critical angle at time $t=0$. 
Assuming that $\coss{\phi}\geq 0$, then there is some $\tau>0$, and a trajectory $\queen(t)=(f(t),g(t))$ of a speed-1 object, where $t\geq 0$, so that 
$
t+\norm{\queen(t) - \ser{}(t)}
$
remains constant, for all $t \in [0,\tau]$. Moreover, $\queen(t)$ can be determined by solving the system of differential equations 
\begin{align}
&\left(f'(t)\right)^2 + \left(g'(t)\right)^2 =1 \label{equa: speed 1}\\
&t+\norm{\queen(t)-\ser{}(t)} = \norm{\ser{}(0)-Q} \label{equa: distance pres}\\
&(f(0),g(0))=(q_1,q_2). \label{equa: initial condition}
\end{align}
\end{theorem}

\begin{proof}%[Proof of Theorem~\ref{thm: cost preservation stationary point}]
An object with trajectory $(f(t), g(t))$ satisfying~\eqref{equa: speed 1} and~\eqref{equa: initial condition} has speed 1 (by Lemma~\ref{lem: unit speed}), and starts from point $Q=(q_1,q_2)$. We need to examine whether we can choose $f,g$ so as to satisfy \eqref{equa: distance pres}.

By Lemma~\ref{thm: cost preservation}, such a trajectory $\queen(t)$ exists exactly when we can guarantee that $\coss{\phi}+\coss{\theta}=1$ over time $t$. When $t=0$ we are given that $\coss{\phi}>0$, hence there exists $\theta$ satisfying $\coss{\phi}+\coss{\theta}=1$. This uniquely determines the velocity of $\queen$ at $t=0$.  

By continuity of the velocities, there must exist a $\tau>0$ such that $\coss{\phi}+\coss{\theta}=1$ admits a solution for $\theta$ also as $\phi$ changes over time $t\in [0,\tau]$, in which time window the cosine of the $(\ser{}, \queen(t))$-critical angle at time $t$ remains non-negative. 
\end{proof}

%We can now prove Theorem~\ref{thm: cost preservation stationary point}

%As it will be clear later, 
Note that condition $\coss{\phi}\geq 0$ of Theorem~\ref{thm: cost preservation stationary point} translates to that $\norm{\ser{}(t)-Q}$  is not increasing at $t=\tau$, i.e. that $\ser{}$ does not move away from point $Q$.

%\paragraph{Improving the Algorithm for $3$ Servants}

Now fix parameters $\alpha,\beta, \rho$ together with the trajectories of $\ser{1}, \ser{2}, \ser{3}$ as in the description of Algorithm $\textsc{Search}_3(\alpha,\beta, \rho)$. The description of our \textit{new algorithm} $\textsc{N-Search}_3(\alpha,\beta, \rho)$ will be complete once we fix a new trajectory for $\queen$. Naming specific values for parameters $\alpha,\beta, \rho$ will eventually prove Theorem~\ref{thm: 3 servants better}.
In order to do so, we introduce some \textit{further notation and conditions}, denoted below by \textit{(Conditions i-iv)},  that we later make sure are satisfied. 

Consider $\queen$'s trajectory as in $\textsc{Search}_3(\alpha,\beta, \rho)$.
Let $\tau_0$ denote a local maximum of 
$$
t+\norm{\queen(t) - \ser{1}(t)}
$$
as it reads for $t\geq 0$ with $|t-1-\alpha| \leq AK(\frac{\alpha+\beta}2,\rho)$ (recall that in this time window, expression is differentiable by Lemma~\ref{lem: cost analysis for Search3}), i.e. 
\begin{equation}
|\tau_0-1-\alpha| \leq AK(\frac{\alpha+\beta}2,\rho) \tag{Condition i}
\end{equation}
%\textit{(Condition i)} $|\tau_0-1-\alpha| \leq AK(\frac{\alpha+\beta}2,\rho)$.
Set $Q=\queen(\tau_0)$, and assume that 
\begin{equation}
\textrm{``The cosine of the $(\ser{},Q)$-critical angle at time $\tau_0$ is non-negative.''} \tag{Condition ii}
\end{equation}
%\textit{(Condition ii)} the cosine of the $(\ser{},Q)$-critical angle at time $\tau_0$ is non-negative.
 Then obtain from Theorem~\ref{thm: cost preservation stationary point} trajectory $(f(t), g(t))$ that has the property that it preserves $\tau_0+\norm{\queen(\tau_0) - \ser{1}(\tau_0)}$ in the time window $[\tau_0, \tau']$. Assume also that
\begin{equation}
\textrm{
``There is time $\tau_1\leq \tau'$ such that point $K_1:=(f(\tau_1), g(\tau_1))$ is equidistant from $\ser{1}(\tau_1), \ser{2}(\tau_1)$,''}
\tag{Condition iii}
\end{equation} 
% \textit{(Condition iii)}  there is time $\tau_1\leq \tau'$ such that point $K_1:=(f(\tau_1), g(\tau_1))$ is equidistant from $\ser{1}(\tau_1), \ser{2}(\tau_1)$ for the first time after time $\tau_0$,
for the first time after time $\tau_0$, such that 
\begin{equation}
\tau_1 \leq 1+\pi-\frac{\alpha+\beta}2. \tag{Condition iv}
\end{equation}
%\textit{(Condition iv)}
%$\tau_1 \leq 1+\pi-\frac{\alpha+\beta}2$. 
Then consider the following modification of $\textsc{Search}_3(\alpha,\beta, \rho)$, where the trajectories of $\ser{1}, \ser{2}, \ser{3}$ remain unchanged, see also Figure~\ref{fig:Search1NEW-S3better}.
\begin{figure}[h!]
\centering
  \includegraphics[width=.61\linewidth]{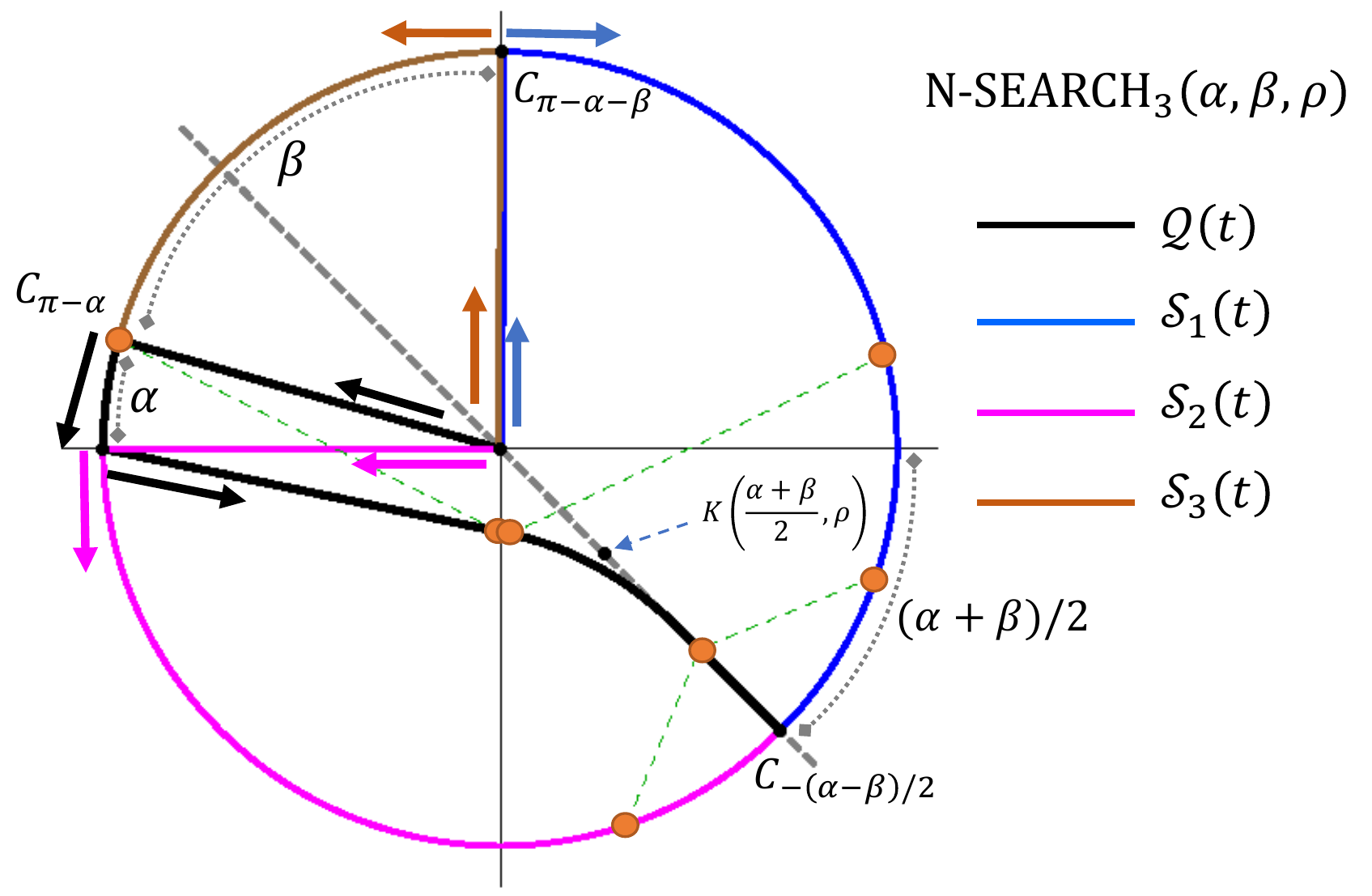}
\caption{
Algorithm $\textsc{Search}_3(\alpha,\beta,\rho)$ depicted for the optimal parameters of the algorithm. 
%Robot's trajectories are also clarified by some velocity vectors (matching the aforementioned colours) at special points.
}
\label{fig:Search1NEW-S3better}
\end{figure}

%\vspace{-0.5cm}
\small{
\begin{center}
\begin{tabular}{|l|}
\hline
\textbf{Algorithm $\textsc{N-Search}_3(\alpha,\beta, \rho)$} \\
\hline \hline
\begin{tabular}{l || lllc}
\textit{Robot}		 	&	\# & \textit{Description} & \textit{Trajectory} & \textit{Duration}\\
\hline 
$\queen$	&	0		&	Move to point $\ki{\pi-\alpha}$										&$\li{O,\ki{\pi-\alpha}, t}$										&	$1$					\\
			&	1		&	Search the circle ccw till point $\ki{\pi}$			&$\ci{\pi-\alpha,t-1}$												&	$\alpha$					\\
			&	2		&	Move toward point $K(\frac{\alpha+\beta}2, \rho)$										&$\li{\ki{\pi}, K(\frac{\alpha+\beta}2,\rho),t-(1+\alpha)}$				&	$\tau_0-1-\alpha$					\\
			&	3		&	Preserve 	$\tau_0+\norm{\queen(\tau_0) - \ser{1}(\tau_0)}$									&  $(f(t),g(t))$ 							&	$\tau_1 - \tau_0$		\\						
			&	4		&	Move to point $\ki{-\frac{\alpha+\beta}2}$										&$\li{K_1,\ki{-\frac{\alpha+\beta}2}}$							&	$\norm{K_1-\ki{-\frac{\alpha+\beta}2}}$ 					\\
\hline 
\end{tabular}
\end{tabular}
\end{center}
}

Note that in phase 2, $\queen$ is not reaching (necessarily) point $K$ rather it moves toward it for a certain duration. 
The search time is still $1+\pi-\frac{\alpha+\beta}2$.
Trajectories of $\ser{1}, \ser{2}, \ser{3}$ are continuous as before, and 
$$\interval{\ser{1}}=
\interval{\ser{2}}=[1, 1+\pi-\frac{\alpha+\beta}2],~ 
\interval{\ser{3}}=[1, 1+\beta],
$$ as well as $\interval{\queen}=[1,1+\alpha]$. 

Condition i makes sure that while $\queen$ is at phase 2, and before it reaches $K(\frac{\alpha+\beta}2,\rho)$, there is a time moment $\tau_0$ when the rate of change of $t+\norm{\queen(t)-\ser{1}(t)}$ is 0. Together with condition ii, this implies that Theorem~\ref{thm: cost preservation stationary point} applies. In fact, for the corresponding critical angles $\phi, \theta$ between $\ser{1}, \queen$ at time $\tau_0$, we have that $\coss{\phi}+\coss{\theta}=1$ by construction. Hence trajectory $(f(t), g(t))$ of phase 3 is well defined, and indeed, $\queen$ jumps from phase 2 to phase 3 while $\queen$ is still moving toward point $K$. Notably, $\queen$'s trajectory is even differentiable at $t=\tau_0$ (but not necessarily at $t=\tau_1$). 
Then, Condition iii says that $\queen$ eventually will enter phase 4, and that this will happen before $\ser{1}, \ser{2}$ finish the exploration of the circle.
Overall, we conclude that in $\textsc{N-Search}_3(\alpha,\rho)$, robots' trajectories are continuous and feasible.
An illustration of the above trajectories for certain values of $\alpha,\beta,\rho$ can be seen in 
Figure~\ref{fig:Search1NEW-S3better}.

%Figure~\ref{fig:Search3better}.
%\begin{figure}[h!]
%  \centering
%  \includegraphics[width=.3\linewidth]{Search3-2Dbetter.pdf}
%\caption{The trajectories of $\queen$ in black, $\ser{1}$ in blue, $\ser{2}$
% in magenta and $\ser{3}$ in brown. }
%\label{fig:Search3better}
%\end{figure}

Now we make some observations, in order to calculate the worst case evacuation time. %as described in Formula~\eqref{equa: total cost}.

\begin{lemma}
\label{lem: cost analysis for Search3better}
Suppose that 
$
\alpha\leq \beta
$,
$1+\beta \leq \tau_0$,
and
$1+\pi-\frac{\alpha+\beta}2 \geq \tau_1+\norm{K_1-\ki{-\frac{\alpha+\beta}2}}$
as well as Conditions i-iv are satisfied. 
Then the following functions are continuous and differentiable in each associated time intervals:
$\norm{\queen(x)-\ser{3}(t)}$ 
in $I_1=\{t\geq 0:~\alpha\leq t-1\leq \beta \}$,
$\norm{\queen(x)-\ser{1}(t)}$ 
in $I_2=\{t \geq 0:~ 1+\alpha\leq t \leq \tau_0$
and in $I_3=\left\{t \geq 0:~ | t - \tau_1|\leq \norm{K_1-\ki{-\frac{\alpha+\beta}2}}\right\}$.
Moreover, the worst case evacuation time of $\textsc{N-Search}_3(\alpha,\beta,\rho)$ can be computed as 
$$
\max
\left\{
\begin{array}{l}
\sup_{t \in I_1 }
\left\{ t+\norm{\queen(t)-\ser{3}(t)} \right\} \\
\sup_{t \in I_2 }
\left\{ t+\norm{\queen(t)-\ser{1}(t)} \right\} \\
\sup_{t \in I_3 }
\left\{ t+\norm{\queen(t)-\ser{1}(t)} \right\} \\
1+\pi-\frac{\alpha+\beta}2
\end{array}
\right\}
$$
\end{lemma}

\begin{proof}%[Proof of Lemma~\ref{lem: cost analysis for Search3better}]
Conditions 
$
\alpha\leq \beta
$ and
$1+\beta \leq \tau_0$
mean that
$\queen$ stops searching no later than $\ser{3}$, 
and that when $\queen$ enters phase 3 after $\ser{3}$ is done searching, 
respectively. 

The line passing through $O$ and $\ki{-(\alpha+\beta)/2}$, call it $\epsilon$, has the property that each point of it, including $K(\frac{\alpha+\beta}2, \rho)$ is equidistant from $\ser{1}, \ser{2}$.
Moreover, while $\ser{1},\ser{2}$ are searching, $\queen$ never goes above line $\epsilon$. Also, while $\queen$ is executing phase 3, $\queen$ remains equidistant from $\ser{1},\ser{2}$ and this is preserved for the remainder of the execution of the algorithm. 
As a result, $\ser{2}$ can be ignored in the performance analysis, and when it comes to the case that $\ser{1}$ finds the exit, the evacuation cost is given by the supremum of  $t+\norm{\queen(t)-\ser{1}(t)}$ in the time interval $I_2$ or in the interval $I_3$. Note that in both intervals, the evacuation cost is continuous and differentiable, by construction. 

If the exit is reported by $\ser{3}$ then the evacuation cost is $t+\norm{\queen(t)-\ser{3}(t)}$ for $t\in [1,1+\beta]$. However, it is easy to see that the cost is strictly increasing for all $t\in [1,1+\alpha]$ (in fact it is linear). Since the evacuation cost is also continuous, we may restrict the analysis in interval $I_1$. 

Lastly, observe that $1+\pi-\frac{\alpha+\beta}2 \geq \tau_1+\norm{K_1-\ki{-\frac{\alpha+\beta}2}}$
implies that $\ser{1},\ser{2}$ reach point $\ki{-(\alpha+\beta)/2}$ no earlier than $\queen$. Hence $\queen$ waits at $\ki{-(\alpha+\beta)/2}$ till the search of the circle is over, which can be easily seen to induce the worse evacuation time after $\queen$ reaches $\ki{-(\alpha+\beta)/2}$.
\end{proof}

Next we prove Theorem~\ref{thm: 3 servants better} by fixing parameters $\alpha,\beta,\rho$ for $\textsc{N-Search}_3(\alpha,\beta,\rho)$.
%, in particular $\alpha=0.27764, \beta=1.29839,\rho=0.68648$.

\begin{proof}[Proof of Theorem~\ref{thm: 3 servants better}]
We choose 
$\alpha=0.27764, 
\beta=1.29839,
\rho=0.68648$.
The trajectories of Figure~\ref{fig:Search1NEW-S3}
%~\ref{fig:Search3better} 
correspond exactly to those values. 
For these values we see that 
$AK(\frac{\alpha+\beta}2,\rho) = 1.29041$,
%AK[aaa, bbb, rrr]
 while $\tau_0-\alpha-1=1.04877$. 
 %ooptimizer
 Hence the transition between phase 1 and phase 2 of $\queen$ is well defined.

The time that 
$\queen$ needs to reach $\ki{-\frac{\alpha+\beta}2}$ 
equals 
$1+\tau_1+\norm{K_1-\ki{-\frac{\alpha+\beta}2}}=3.18073$, 
% 1 + TimeToResumeToLineOfSymmetry +  Norm[ NEWqueen[1 + TimeToResumeToLineOfSymmetry] -  cycle[-(aaa + bbb)/2]]
while the time that $\ser{1}, \ser{2}$ reach the same point is 
$1+\pi-\frac{\alpha+\beta}2 =3.35358$. 
%1 + Pi - (aaa + bbb)/2
Therefore we may attempt to solve numerically the differential equation of Theorem~\ref{thm: cost preservation stationary point}. It turns out that for the resulting trajectory $(f(t),g(t)$, and for $\tau_1=2.89288$, 
%1 + TimeToResumeToLineOfSymmetry
point $(f(\tau_1),g(\tau_1)$ is equidistant from $\ser{1}, \ser{2}$. 
Moreover, $\queen$ enters phase 4 at time $\tau_1=2.89288$, prior to $1+\pi-\frac{\alpha+\beta}2$.
Hence, Conditions i-iv are all met, as well as Lemma~\ref{lem: cost analysis for Search3better} applies.

From the above, it is immediate that the worst evacuation time after $\queen$ reaches $\ki{-(\alpha+\beta)/2}$ equals $1+\pi-\frac{\alpha+\beta}2 =3.35358$.
%1 + Pi - (aaa + bbb)/2
Hence, it remains to compute the maxima of $t+\norm{\queen(t)-\ser{3}(t)}$ in interval $I_1$, and of $t+\norm{\queen(t)-\ser{1}(t)}$ in intervals $I_2$, $I_3$. 

To that end, when $t \in I_1=[1.27764,2.29839]$ we have that 
%1+aaa
%1+bbb
\begin{align*}
\queen(t) &=\left(0.978782 t-2.25053,0.261795\, -0.204905 t \right) \\
\ser{3}(t) &=\left(\cos (t+0.565563),\sin (t+0.565563)\right), 
%Expand[linePoints[pointA, pointK[aaa, bbb, rrr], t - aaa - 1]]
%cycle[Pi - aaa - bbb + t - 1]
\end{align*}
so that $t+\norm{\queen(t)-\ser{3}(t)}$ becomes
$$
t+
\sqrt{(-0.204905 t-\sin (t+0.565563)+0.261795)^2+(0.978782 t-\cos (t+0.565563)-2.25053)^2}$$
%t + newnorm[  Expand[linePoints[pointA, pointK[aaa, bbb, rrr], t - aaa - 1]] -    cycle[Pi - aaa - bbb + t - 1]]
in which case
$$
\sup_{t \in I_1 }
\left\{ t+\norm{\queen(t)-\ser{3}(t)} \right\}
=
1+\beta+
\norm{\queen(1+\beta)-\ser{3}(1+\beta)}
\approx 3.37387
%dq3[aaa, bbb, rrr, 1 + bbb]
$$

When $t \in I_2=[1.27764,2.32641]$, 
%1 + aaa
%1 + aaa + ooptimizer
$\queen$'s trajectory is the same as in $I_1$ and 
\begin{align*}
\ser{1}(t) =\left(\cos (2.56556\, -t),\sin (2.56556\, -t)\right), 
%cycle[Pi - aaa - bbb - t + 1]
\end{align*}
so that $t+\norm{\queen(t)-\ser{1}(t)}$ becomes
$$
t+
\sqrt{(-\sin (2.56556\, -t)-0.204905 t+0.261795)^2+(-\cos (2.56556\, -t)+0.978782 t-2.25053)^2}.
%d1better[t_] :=  t + newnorm[   Expand[linePoints[pointA, pointK[aaa, bbb, rrr], t - aaa - 1]] -     cycle[Pi - aaa - bbb - t + 1]]
%    d1better[t]
$$

When $t \in I_3=[2.89288,3.18073]$, $\ser{1}$'s trajectory is the same as in $I_2$ and 
%1 + TimeToResumeToLineOfSymmetry
%1 + TimeToResumeToLineOfSymmetry +  Norm[ NEWqueen[1 + TimeToResumeToLineOfSymmetry] -  cycle[-(aaa + bbb)/2]]
\begin{align*}
\queen(t) &=\left(0.705254 t-1.53797,1.54604\, -0.708955 t
0.706399 t-1.53762,1.5407\, -0.707814 t\right), 
%Expand[linePoints[NewPointOnLineOfSymmetry, pointC[aaa, bbb],   t - 1 - TimeToResumeToLineOfSymmetry]]
\end{align*}
so that $t+\norm{\queen(t)-\ser{1}(t)}$ becomes
$$
t+
\sqrt{(-\sin (2.56556\, -t)-0.708955 t+1.54604)^2+(-\cos (2.56556\, -t)+0.705254 t-1.53797)^2}.
%d1better22[t_] :=  t + newnorm[   Expand[linePoints[NewPointOnLineOfSymmetry, pointC[aaa, bbb],       t - 1 - TimeToResumeToLineOfSymmetry]] -     cycle[Pi - aaa - bbb - t + 1]]
%d1better22[t]
$$

Numerically 
\begin{align*}
\sup_{t \in I_2 }
\left\{ t+\norm{\queen(t)-\ser{1}(t)} \right\}
&=
 \tau_0+\norm{\queen(\tau_0)-\ser{1}(\tau_0)} 
 = \tau_1+\norm{\queen(\tau_1)-\ser{1}(\tau_1)} \\
&=\sup_{t \in I_3 }
\left\{ t+\norm{\queen(t)-\ser{1}(t)} \right\} 
\approx
3.37387. 
\end{align*}

The reader may also consult Figure~\ref{fig:Search1NEW-S3better}.
%\ref{fig:Search3better} 
%that depicts the locations of the optimizers, i.e the worst case locations on the circle for the exit to be found by any of the robots. This Figure also shows corresponding locations of $\queen$ on her trajectory, along with the corresponding evacuation trajectory in dashed green colour.
\end{proof}

\section{Lower Bounds}
\label{seclbounds}

In this section we derive lower bounds for evacuation. In Section~\ref{sec:lb1} we treat the case of $n=1$ (see Theorem~\ref{thm:lb2}) and in Section~\ref{sec:lb23} we treat the case of $n=2$ and $3$ (see Theorem~\ref{thm:lb23}).
% Our results are summarized as follows. 
%We summarize the main results of the section in Theorem~\ref{thm:lb_123}.
\ignore{
\begin{theorem}\label{thm:lb_123}
For $n=1,2,3$, any evacuation algorithm for \pe{n} requires time at least $\IT_n$, where $\IT_1 \approx 4.3896$, $\IT_2 \approx 3.6307$, and $\IT_3 \approx 3.2017$.
\end{theorem}
}
\subsection{Lower Bound for \pe{1}}\label{sec:lb1}
    We will derive the lower bound using an adversarial argument placing the exit at an unknown vertex of a regular hexagon.

    \begin{theorem}
    \label{thm:lb2}
    The worst-case evacuation time for \pe{1} is at least $3 + \pi/6 + \sqrt{3}/2 \approx 4.3896$
    \end{theorem}

    \begin{proof}%[Proof of Theorem~\ref{thm:lb2}]
    At time $1 + \pi/6$, at most $\pi/3$ of the perimeter of the circle can have been explored by the queen and servant. Thus, there is a regular hexagon, none of whose vertices have been explored. If the exit is at one of these vertices, by Theorem~\ref{lower-bound-hexagon}, it takes $2 + \sqrt{3}/2$ for the queen to evacuate. The total time is $1 + \pi/6 + 2 + \sqrt{3}/2$.
    %\qed
    \end{proof}

    %\paragraph{Lower bound on a unit-side hexagon.}
    %\section{Lower bound for the queen problem on  a unit-side hexagon}

    Next we proceed to provide a lower bound on a unit-side hexagon. Label the vertices of the hexagon $V$ as $A, \ldots, F$ as shown in Figure~\ref{fig:lb-hexagon-left}.
    \begin{figure}
    \begin{center}
    \includegraphics[height=1.9in]{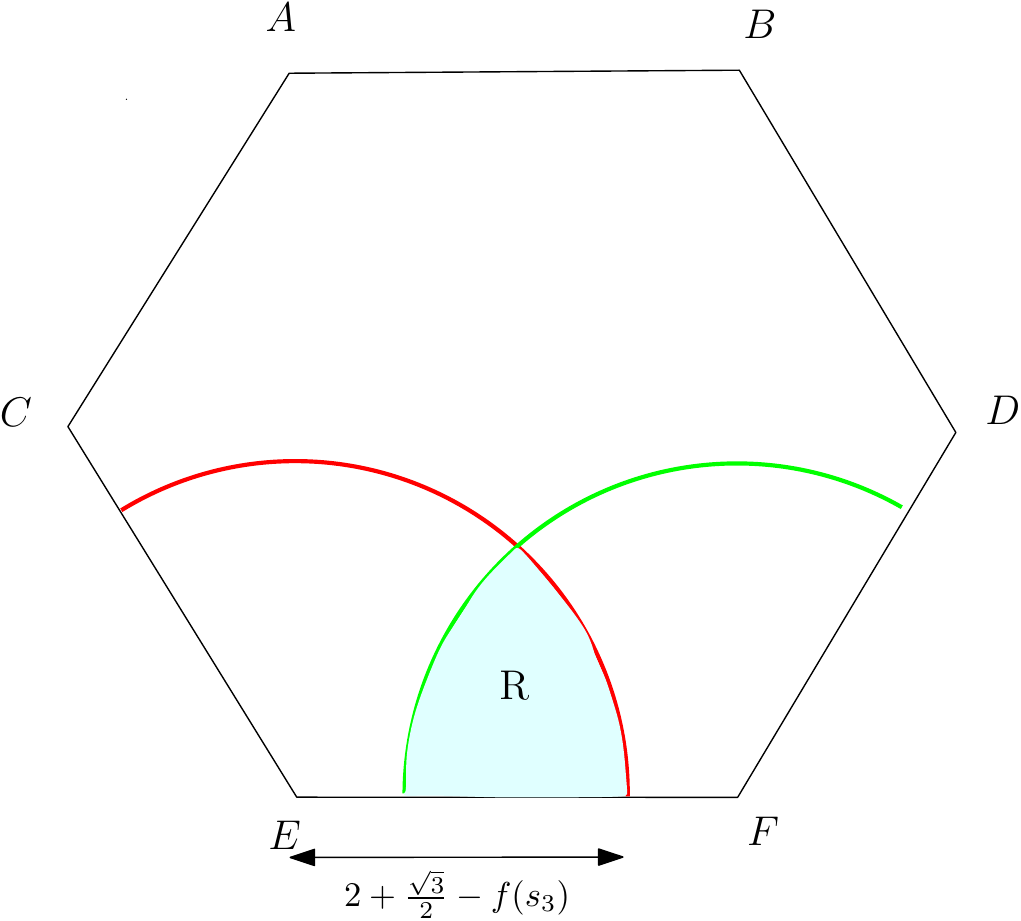}
    \end{center}
    \caption{(Left) The queen must be in region $R$ at time $f(s_3)$. Here $s_3=E$ and $q_3 = F$.}
    \label{fig:lb-hexagon-left}
    \end{figure}
    Fix an evacuation algorithm $\cal A$. For any vertex $v$ of the hexagon, we call $f(v)$ the time of {\em first visit} of the vertex $v$ by either the servant or the queen, according to algorithm $\cal A$. We call $q(v)$ the time that the queen gets to the vertex $v$. Clearly, $q(v) \geq f(v)$, and if the queen arrives at the vertex no later than the servant, $q(v) = f(v)$. 
    %\marginpar{Statement of Theorem~\ref{lower-bound-hexagon} not clear.}

    \begin{theorem}
    \label{thm:lb3}
        For any algorithm $\cal{A}$, the evacuation time for the queen when the exit is at one of the vertices of the hexagon is $max_{v \in V}\{ q(v) \} \geq 2 + \sqrt{3}/2$.
        \label{lower-bound-hexagon}
    \end{theorem}

    \begin{proof}%[Proof of Theorem~\ref{thm:lb3}]
    Suppose there is an algorithm in which the queen can always evacuate in time $< 2 + \sqrt{3}/2$. Consider the trajectories of the servant and the queen. If either the queen or the servant are the first to visit 4 vertices, then for the fourth such vertex $v$, we have $f(v) \geq 3$, a contradiction. Therefore, the queen is the first to visit three vertices, and the servant is the first to visit three vertices. We denote the three vertices visited first by the servant as $s_1, s_2, s_3$ (in the order they are visited) and the three vertices visited first by the queen as $q_1, q_2, q_3$, and note that they are all distinct. 

    Notice that neither $s_3$ nor $q_3$ can be visited before time 2, that is,  $f(s_3), f(q_3)  \geq 2$. If  $f(q_3) \leq f(s_3)$, then we place the exit at $s_3$, and the queen needs time at least 1 to get to $s_3$,  which implies that $T \geq q(s_3) \geq f(q_3) + 1 \geq 3$, a contradiction. We conclude that at time $f(s_3)$, the queen is yet to visit $q_3$. Since the exit can be at either $s_3$ or $q_3$, at time $f(s_3)$, the queen must be at distance  $< 2 + \sqrt{3}/2 - f(s_3) \leq \sqrt{3}/2$ from {\em both} $s_3$ and $q_3$.

    Assume without loss of generality that $s_3 = E$ (see Figure~\ref{fig:lb-hexagon-left}). 
    Since $A, B, D$ are all at distance at least $\sqrt{3}$ from $E$, we conclude that $q_3$ is either $C$ or $F$. Assume without loss of generality that $q_3 = F$. Let $R$ denote the lens-shaped region that is at distance $< 2 + \sqrt{3}/2 - f(s_3)$ from both $E$ and $F$. Recall that at time $f(s_3)$, the queen must be inside the region $R$. Notice that if $f(s_3) \geq 1.5 + \sqrt{3}/2$, the region $R$ is empty, yielding a contradiction. So it must be that $2 \leq f(s_3) < 1.5 + \sqrt{3}/2.$

\ignore{
    \begin{figure}
    \begin{center}
    \includegraphics[height=1.7in]{new-hexagon.pdf}
    ~~\includegraphics[height=1.7in]{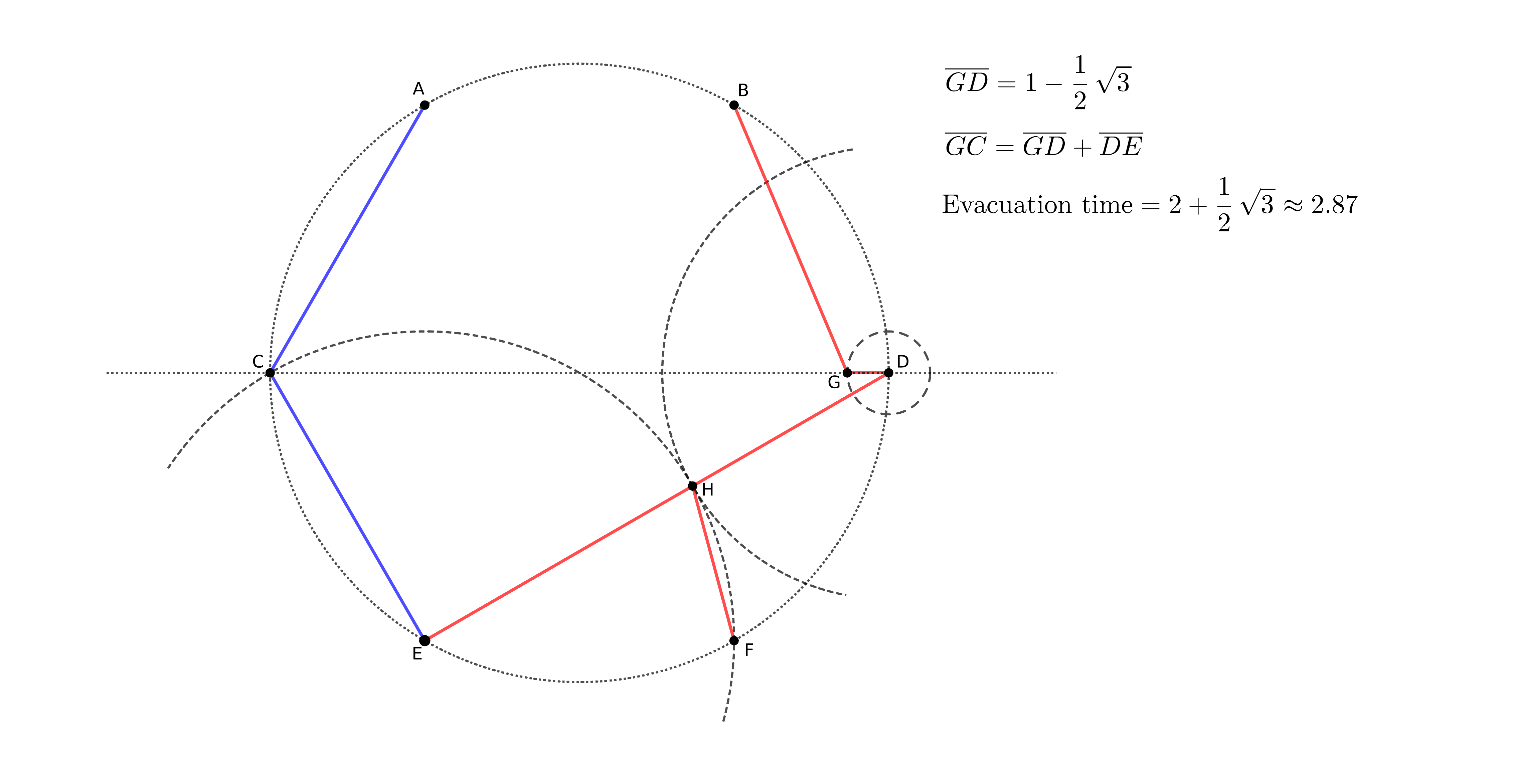}
    \end{center}
    \caption{(Left) The queen must be in region $R$ at time $f(s_3)$. Here $s_3=E$ and $q_3 = F$. (Right) Blue trajectory: servant and red trajectory: queen. At point $H$, if the queen hears of an exit at $E$, she goes there, otherwise she goes to $F$.}
    \label{fig:lb-hexagon}
    \end{figure}
}

    We now work backwards to deduce the trajectories of the servant and the queen. Clearly $s_2 \neq F$ since $q_3 = F$. If $s_2 \neq C$, then $f(s_3) \geq \sqrt{3} + 1 > 1.5+ \sqrt{3}/2$, a contradiction. Therefore, $s_2 = C$. By the same reasoning, $s_1 = A$. Therefore, the queen is the first to visit $D$ and $B$. If $q_1=D$ and $q_2=B$, we place the exit at $E$; since $f(q_2) \geq 1$ and $dist(B, E) =2$, we have $T \geq q(E)  \geq 3$, a contradiction. Thus, $q_2 = D$ and $q_1 = B$.  

    Consider the location of the queen at time 1. If she is at distance $\geq 1 + \sqrt{3}/2$ from $C$ at time $1$, then if the exit is at $C$, $q(C) \geq 2 + \sqrt{3}/2$. So at time 1, the queen must be at distance $< 1 + \sqrt{3}/2$ from $C$ and consequently she is at distance $\geq 1 - \sqrt{3}/2$ from vertex $D$. Therefore $f(q_2) = f(D) \geq 2 - \sqrt{3}/2$. Also, $f(D) < 1.5$  since if the queen reaches $D$ at or after time 1.5, she cannot reach the region $R$ before time $1.5 + \sqrt{3}/2 > f(s_3)$. So $f(D) \leq f(s_3)$. If the exit is at $E = s_3$, the queen cannot reach the exit before time $f(D) + dist(D, E) \geq 2 - \sqrt{3}/2 + \sqrt{3}  = 2 + \sqrt{3}$, concluding the proof by contradiction.
    \end{proof}

    We remark that the above bound is optimal, and is achieved by the algorithm depicted in Figure~\ref{fig:lb-hexagon-right}.
    \begin{figure}
    \begin{center}
\includegraphics[height=1.9in]{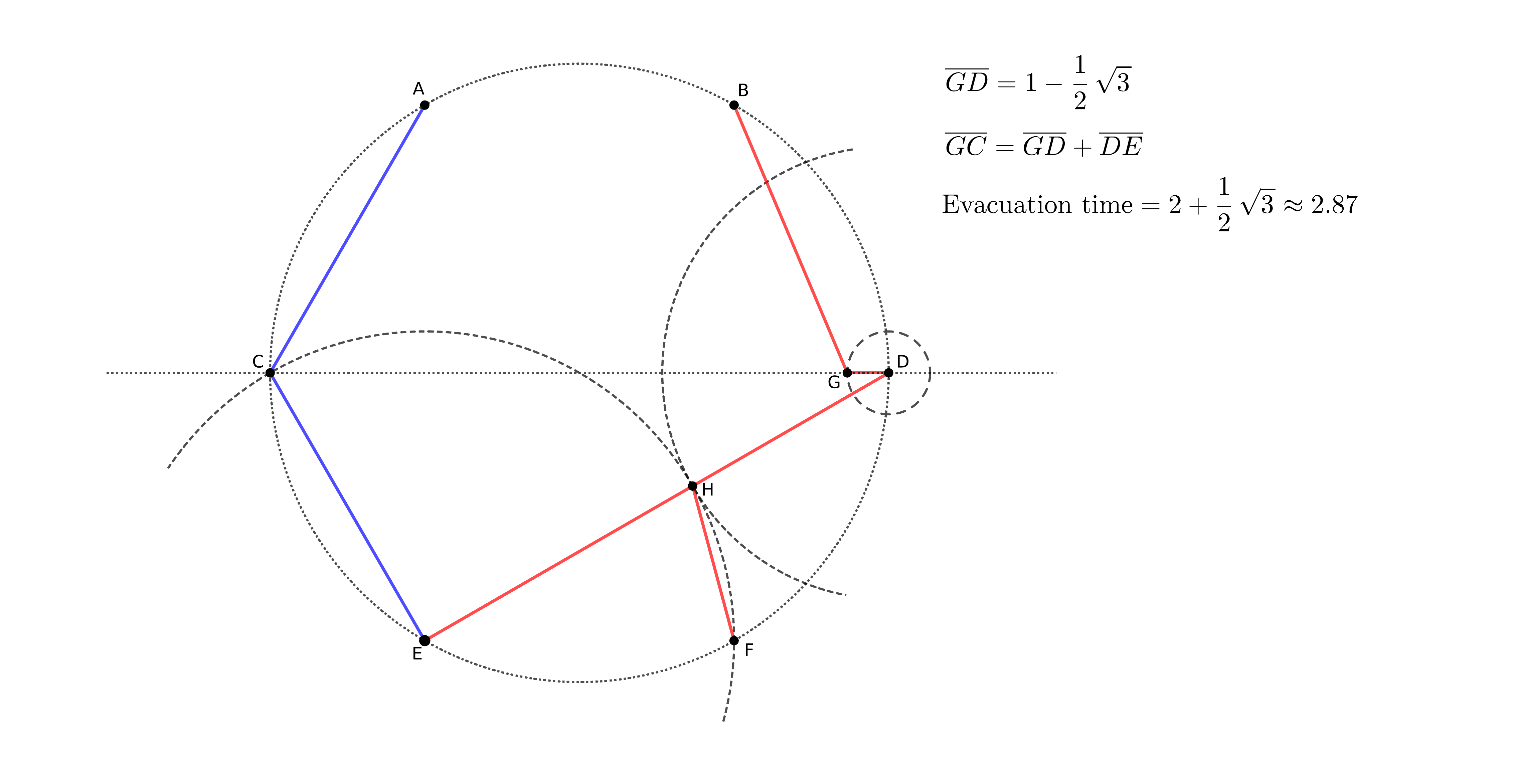}
    \end{center}
    \caption{Blue trajectory: servant and red trajectory: queen. At point $H$, if the queen hears of an exit at $E$, she goes there, otherwise she goes to $F$.}
    \label{fig:lb-hexagon-right}
    \end{figure}

\subsection{Lower Bounds for \pe{2} and \pe{3} - Proof Outline}\label{sec:lb23}
    In the case of $n=2$ and $n=3$ the proof is rather technical. Next we present a high level outline as to why the lower bounds hold. 
    %The full details of the proof are provided in the appendix. 
    
    \begin{theorem}
    \label{thm:lb23}
    The worst-case evacuation time for \pe{2} is at least $3.6307$
    and for \pe{3} at least $3.2017$.
    \end{theorem}

    Throughout this section we will use $\IT$ to refer to the evacuation time of an arbitrary algorithm and use $\UC$ to refer to the unit circle which must be evacuated.

    The main thrust of the proof relies on a simple idea -- the queen should aid in the exploration of $\UC$. This is immediately evident for the particular case of $n=2$ since, if the queen does not explore, it will take time at least $1+\pi$ for the servants to search all of $\UC$ and we already have an upper bound smaller than this (Theorem~\ref{thm: 2 servants}). Thus, a general overview of the proof is as follows: we show that in order to evacuate in time $\IT$ the queen must explore some minimum length of the perimeter of $\UC$. We will then demonstrate that the queen is not able to explore this minimum amount in any algorithm with evacuation time smaller than what is given in Theorem~\ref{thm:lb23}.

    To be concrete, consider the case of $n=2$ and assume that we have an algorithm with evacuation time $\IT < 1+\pi$. Then, in order for the robots to have explored all of $\UC$ in time $\IT$, the queen must explore a subset of the perimeter of total length at least $2(1+\pi-\IT)$. Intuitively, this minimum length of perimeter will increase in size as $\IT$ decreases.

    Now consider that it is not possible for the queen to always remain on the perimeter (indeed, in each of the algorithms presented, the queen leaves the perimeter). To see why this is consider that, in any algorithm with evacuation time $\IT$, it must be the case that all unexplored points of $\UC$ are located a distance no more than $\IT-t$ from the queen at all times $t \leq \IT$. If the queen is on the perimeter at any time $t$ satisfying $\IT-t \leq 2$, then, there will be some arc $\theta(t,\IT) \subset \UC$ (see Lemma~\ref{lm:theta}) such that all points of $\theta(t,\IT)$ are at a distance at least $\IT-t$ from the queen. Thus, if the queen is to be on the perimeter at the time $t$ we can conclude that all of the arc $\theta(t,\IT)$ must have already been discovered. However, we will find (see Lemma~\ref{lm:dtheta}) that $\theta(t,\IT)$ will often grow at a rate much larger than the robots can collectively explore and at some point the queen will have to leave the perimeter. In fact, there will be an interval of time during which it is not possible for the queen to be exploring and this in turn implies that there is a maximum amount of perimeter that can be explored by the queen. Intuitively, the maximum length of perimeter that can be explored by the queen will decrease as $\IT$ decreases. The lower bound will result by balancing the minimum amount of perimeter the queen needs to search and the maximum amount of perimeter that the queen is able to search.

    The above argument will need a slight modification in the case of $n=3$. In this case we will show that there is some critical time $t_*$ before which the queen must have explored some minimum amount of perimeter. Again, the lower bound follows by balancing the maximum amount of perimeter the queen can explore by the time $t_*$ and the minimum amount of perimeter the queen needs to explore before the time $t_*$.

%%%%%%%%%%%%%%%%%%%end of lower bounds
%%%%%%%%%%%%%%%%%%%%%%%%%%%%%
%%%%%%%%%%%%%%%%%%%%%%%%%%%%%

\subsection{Lower Bounds for \pe{2} and \pe{3} - Proof Details}
    In this section we present the complete details of the proofs for the lower bounds in the cases $n=2$ and $n=3$. Throughout this section we will use $\IT$ to refer to the evacuation time of an arbitrary algorithm and use $\UC$ to refer to the unit circle which must be evacuated. 
    
    The idea of the proofs are to bound the amount of perimeter the queen can search for a given evacuation time $\IT$ and then show that the queen must search a minimum amount of the perimeter in order to achieve the evacuation time $\IT$. The lower bounds result by balancing the minimum amount of perimeter the queen must search with the maximum amount of perimeter the queen can search.

    We begin with two lemmas which will be used for both the $n = 2$ and $n=3$ bounds. Their necessity will become apparent shortly.
    \begin{lemma} \label{lm:theta}
        Consider any $r < 2$ and a point $P \in \UC$. Define the circle $\ID_P$ as the disk centered on $P$ with radius $r$. Then the subset of the perimeter of $\UC$ which is not contained in $\ID_P$ has length $\theta = 4\acoss{\frac{r}{2}}$.
    \end{lemma}

    \begin{proof}
        Without loss of generality assume that the point $P$ is located at $\pair{-1}{0}$. Since $r < 2$ the disks $\UC$ and $\ID_P$ will intersect at two boundary points $A$ and $B$ between which the distance along the perimeter of $\UC$ is $\theta$. This situation is depicted in Figure~\ref{fig:lb1}. Referring to this figure, one can easily observe that $r = 2\sinn{\frac \pi 2 - \frac {\theta}{4}} = 2\coss{\frac{\theta}{4}}$. Rearranging for $\theta$ we find that $\theta = 4 \acoss{\frac{r}{2}}$.
        \includeFig{width=2in,keepaspectratio}{fig:lb1}{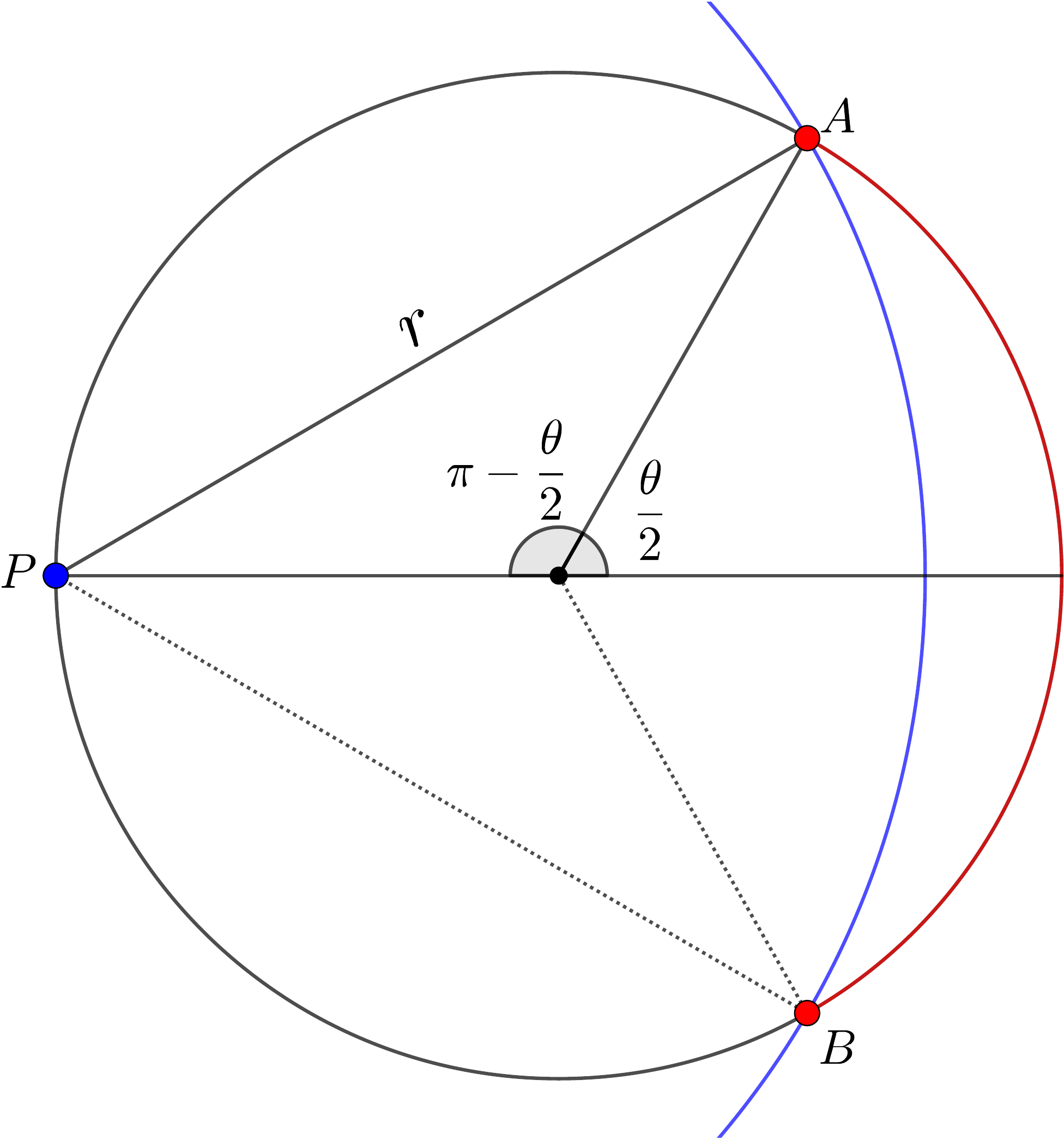}{Setup for the proof of Lemma~\ref{lm:theta}. The boundary of the disk $\ID_P$ is indicated in blue. The arc of $\UC$ which is excluded from $\ID_P$ is highlighted in red and has length $\theta$.}
    \end{proof}

    \begin{lemma}\label{lm:dtheta}
        Consider the function $\theta(t,\IT) = 4 \acoss{\frac{\IT-t}{2}}$ with $\IT > 0$. Then $\diff{\theta}{t} > 2$ for all $t$ satisfying $\IT-2 < t < \IT$ and $\diff{\theta}{t} > 3$ for $t$ satisfying $\IT-2 < t < \IT-\frac 23 \sqrt{5}$. Furthermore, $\diff{\theta}{\IT} < -2$ for all $\IT-2 < t < \IT$.
    \end{lemma}

    \begin{proof}
        The rate of change of $\theta(t,\IT)$ with $t$ is given by
        \[\diff{\theta}{t} = \frac{4}{\sqrt{4-(\IT-t)^2}}.\]
        From this relation it is simple to confirm that $\diff{\theta}{t} > 2$ for $\IT-2 < t < \IT$ and that $\diff{\theta(r)}{r} > 3$ for $\IT-2 < t < \IT-\frac 23 \sqrt{5}$. It should also be obvious by the symmetry of $\IT$ and $t$ in the function $\theta(t,\IT)$ that $\diff{\theta}{\IT} < -2$ for all $\IT-2 < t < \IT$.
    \end{proof}

    \subsubsection{Lower bound for $n=2$}
    We begin with the main result of the section.
    \begin{theorem}\label{thm:lb_n2}
        For $n=2$ and any algorithm the queen cannot be evacuated in time less than $\IT_2$ which is the solution to the equations
        \[\tau = \IT_2 - 2\coss{\frac{\tau-1}{2}}\]
        \[t_* = \frac 12 (\IT_2+1)\]
        \[\IT_2 = t_* + 2\coss{\frac{2t_*+\tau}{4} - \frac 34}.\]
        Solving these equations numerically gives $\tau \approx 1.7815$, $t_* \approx 2.3154$, and $\IT_2 \approx 3.6307$.
    \end{theorem}
    We will see that the queen cannot be located on the perimeter of the circle during the interval of time $(\tau,\ t_*)$ and thus $\tau-1$ represents the maximum amount of perimeter that can be explored by the queen before the time $t_*$. The time $t_*$ is chosen such that for all $\IT < \IT_2$ a solution to the equations in Theorem~\ref{thm:lb_n2} do not exist, and, as such, $\tau-1$ will represent the maximum length of the perimeter that can be explored by the queen. In the following lemma we show that the queen must explore a length of the perimeter greater than $\tau-1$ in order to evacuate in time less than $\IT_2$.
    \begin{lemma}\label{lm:lb2_1}
        For $n = 2$ and any evacuation algorithm with $\IT < 1+\pi$, the queen must explore a subset of the perimeter of length $y \geq 2(1 + \pi - \IT)$. In particular, if $\IT < \IT_2$, we need $y > 2(1 + \pi - \IT_2) \approx 1.0217$.
    \end{lemma}
    \begin{proof}
        If the queen explores a subset of the perimeter of length $y$ then the robots will take time $1+\frac{2\pi-y}{2}$ to explore the circle. The robots need to at least explore the entire circle in time $\IT$ and therefore $1+\frac{2\pi-y}{2} \leq \IT$, or, equivalently, $y \geq 2(1 + \pi - \IT_2)$. For $\IT < \IT_2 \approx 3.6307$ we need $y > 1.0217$.
    \end{proof}

    We will now show that the maximum length of perimeter the queen can explore is less than $\tau-1$ if $\IT < \IT_2$. This will be the goal of the next two lemmas.
    \begin{lemma}\label{lm:impl1}
        Consider the equation $\IT = t + 2\coss{\frac 12 (t-1)+\frac 12 \alpha}$ with $\IT > 0$, $\alpha$ satisfying $0 < \alpha \leq t$ and $t$ satisfying $1 < t \leq \IT$. Then $\diff{t}{\IT} \geq \frac 12$, and, if $0 < t < 1+2\pi-\frac{\alpha}{2}$ then $\diff{t}{\alpha} > 0$.
    \end{lemma}

    \begin{proof}
        Implicitly differentiating the equation $\IT = t + 2\coss{\frac 12 (t-1)+\frac 14 \alpha}$ with respect to $\IT$ gives us
        \[\diff{t}{\IT} = \frac{1}{1-\sinn{\frac 12 (t-1) + \frac 14 \alpha}}.\]
        Since the sine function ranges from $-1$ to $1$ we can easily see that $\diff{t}{\IT} \geq \frac 12$.

        Implicitly differentiating the equation $\IT = t + 2\coss{\frac 12 (t-1)+\frac 14 \alpha}$ with respect to $\alpha$ gives us
        \[\diff{t}{\alpha} = \frac 12 \cdot \frac{\sinn{\frac 12 (t-1) + \frac 14 \alpha}}{1-\sinn{\frac 12 (t-1) + \frac 14 \alpha}}.\]
        We can esily see that the denominator of $\diff{t}{\alpha}$ will never be negative and thus $\diff{t}{\alpha} > 0$ provided that the numerator is positive. This clearly occurs for $\frac 12 (t-1) + \frac 14 \alpha < \pi$ or $t < 1 + 2\pi - \frac \alpha 2$.
    \end{proof}

    \begin{lemma}\label{lm:lb2_2}
        Define $\tau$ as in Theorem~\ref{thm:lb_n2}. Then, for $n = 2$ and any evacuation algorithm with $\IT < \IT_2$, the queen cannot explore a subset of the perimeter with length $y > \tau-1$.
    \end{lemma}

    \begin{proof}
        We start with an observation: if the queen is to evacuate in time $\IT$, then, at any time $t < \IT$, all points of $\UC$ that are a distance greater than $\IT-t$ from the queen must be explored by a robot. If the queen is located on the perimeter at the time $t > \IT-2$ then by Lemma~\ref{lm:theta} there is an arc of length
        \begin{equation*}
            \theta(t,\IT) = 4 \acoss{\frac{\IT-t}{2}}
        \end{equation*}
        all points of which lie a distance greater than $\IT-t$ from the queen (as an abuse of notation we will refer to the arc with length $\theta(t,\IT)$ as $\theta(t,\IT)$). Thus, in order for the queen to be on the perimeter at the time $t$, the arc $\theta(t,\IT)$ must be explored. As we have $3$ robots in total the maximum length of $\theta(t,\IT)$ that can be explored at any time $t$ is $3(t-1)$. However, we claim that the queen cannot have explored any of $\theta(t,\IT)$ if the time $t$ satisfies $t < \frac 12 (\IT+1)$. Indeed, observe that the endpoints of $\theta(t,\IT)$ lie a distance $\IT-t$ away from the queen (by definition) and the queen -- who took a unit of time to reach the perimeter -- could have explored a point on the perimeter at most a distance $t-1$ from her current position. Thus, if $t-1 < \IT-t$, or, alternatively, $t < \frac 12 (\IT+1)$, the queen cannot have explored any of the arc $\theta(t,\IT)$. We must therefore have $\theta(t,\IT) \leq 2(t-1)$ for times $t$ that satisfy $t < \frac 12(\IT+1)$.

        We note that there is a trivial lower bound of $1+\frac{2\pi}{3} > 3$ and thus we can assume that $\IT > 3$. We make the following claim: if $\IT < \IT_2$ then the smallest time $t_0 > 0$ solving $\theta(t_0,\IT) = 2(t_0-1)$ satisfies $\left.\diff{\theta}{t} \right|_{t=t_0} > 2$ and $t_0 < \frac 12(\IT+1)$. We note that, if this is the case, the queen will have to leave the perimeter at the time $t_0$  (since she has not explored any of the arc $\theta(t,\IT)$ and, immediately after the time $t_0$, $\theta(t,\IT)$ will be too large to have been explored by the servants alone). 
        
        We first show that $t_0 < \frac 12 (\IT+1)$. To this end we rearrange the equation $\theta(t_0) = 2(t_0-1)$ to get
        \[t_0 = \IT - 2\coss{\frac{t_0-1}{2}}\]
        which is the definition of $\tau$ in Theorem~\ref{thm:lb_n2} (in the case that $\IT = \IT_2$). One can easily confirm that in the case of $\IT = \IT_2$ we have $\left.\diff{\theta}{t} \right|_{t=\tau} \approx 5.2511 > 2$ and $\tau < \frac 12(\IT+1)$. Now observe that $\theta(t,\IT)$ is a decreasing function of $\IT$ and this implies that for $\IT<\IT_2$ we have $\theta(\tau,\IT) > \theta(\tau,\IT_2)$. We can therefore conclude that the time $t_0$ must occur earlier than the time $\tau$. We note that $\tau < 2$ and, since $\IT \geq 3$, we have $\tau < \frac 12 (\IT+1)$. Since $t_0 < \tau$ we can conclude that $t_0 < \frac 12 (\IT+1)$. 
       
        The second part of the claim follows directly from Lemma~\ref{lm:dtheta} where we show that $\diff{\theta}{t} > 2$ for all $t$ satisfying $\IT-2 < t < \IT$.

        As the queen must leave the perimeter at the time $t_0 < \tau$, by Lemma~\ref{lm:lb2_1}, we can say that the queen must be able to return to the perimeter and explore before the algorithm terminates. Thus, consider the smallest time $t_1 > t_0$ at which the queen may return to the perimeter. In order for the queen to be on the perimeter we will still need the arc $\theta(t,\IT)$ to be completely explored. However, in this case it may be possible that $t_1 \geq \frac 12 (\IT+1)$ and as such the queen could have explored at most a length $t_0-1$ of $\theta(t,\IT)$ at the time $t_1$. We can therefore conclude that $t_1$ will satisfy $\theta(t_1) = 2(t_1-1)+y$ with $y = 0$ if $t_1 < \frac 12 (\IT+1)$, and $y \leq t_0-1$ if $t_1 \geq \frac 12 (\IT+1)$. Writing the equation $\theta(t_1) = 2(t_1-1)+y$ in full and rearranging we find that
        \[t_1 = \IT - 2\coss{\frac 12(t_1-1)+\frac 14 y}.\]
        We will now consider the cases $t_1 < \frac 12(\IT+1)$ and $t_1 \geq \frac 12(\IT+1)$ separately.\\

        \noindent \textbf{Case 1: $t_1 < \frac 12 (\IT+1)$}\\
            In this case $t_1$ can be observed to satisfy the same equation as $t_0$. We claim that this is not possible if $t_1 > t_0$. Indeed, by Lemma~\ref{lm:dtheta} we have $\diff{\theta}{t} > 2$ and the arc $\theta(t,\IT)$ will always grow at a rate larger than the servants alone can explore. Thus, a solution to the equation $\theta(t_1) = 2(t_1-1)$ with $t_1 > t_0$ does not exist. This implies that the queen can explore a maximum subset of the perimeter of total length $t_0-1 < \tau-1$ if $t_1 < \frac 12 (\IT+1)$.\\

        \noindent \textbf{Case 2: $t_1 \geq \frac 12 (\IT+1)$}\\
            In this case $t_1$ satisfies
            \[t_1 = \IT - 2\coss{\frac 12 (t_1-1) + \frac 14 y}.\]
            Although it can be confirmed that $\diff{t_1}{y} > 0$ (see Lemma~\ref{lm:impl1}) we will show that, even when $t_1$ is as large as possible (i.e. $y = t_0-1$), we cannot have $t_1 \geq \frac 12(\IT+1)$. Thus we assume that $t_1$ satisfies
            \[t_1 = \IT - 2\coss{\frac 12 (t_1-1) + \frac 14 (t_0-1)}.\]
            Now write $t_1 = t_1(\IT)$ as a function of $\IT$ and note that, by Lemma~\ref{lm:impl1}, we have $\diff{t_1}{\IT} > \frac 12$. Using this we can say that $t_1(\IT_2) - t_1(\IT) > \frac 12(\IT_2-\IT)$. By definition of $\IT_2$ we have $t_1(\IT_2) = \frac 12 (\IT_2+1)$ and we can therefore write $\frac 12 (\IT_2+1) - t_1(\IT) > \frac 12 (\IT_2-\IT)$. Rearranging this inequality gives us $t_1(\IT) < \frac 12 (\IT+1)$ which contradicts with our assumption that $t_1 \geq \frac 12 (\IT+1)$ and we must conclude that $t_1 < \frac 12 (\IT+1)$. This concludes the proof.
        \end{proof}

        At this point the proof of Theorem~\ref{thm:lb_n2} is rather trivial.
        
\begin{proof}(Theorem~\ref{thm:lb_n2})\\
            Assume that we have an algorithm with evacuation time $\IT < \IT_2$. Then, by Lemma~\ref{lm:lb2_1}, the queen must explore a subset of the perimeter of length at least $y > 1.0217$. However, by Lemma~\ref{lm:lb2_2}, the queen can only explore a subset of the perimeter of length $y < \tau-1 \approx 0.7815$ if $\IT < \IT_2$. It is therefore not possible for the queen to evacuate in time less than $\IT_2$.
        \end{proof}

    \subsubsection{Lower bound for $n=3$}
    The main result of this section is given below:
    \begin{theorem}\label{thm:lb_n3}
        For $n=3$ and any algorithm the queen cannot be evacuated in time less than $\IT_3$ which is the solution to the equations
        \[\tau = \IT_3 - 2\coss{\frac 34 (\tau-1)}\]
        \[t_* = 1 + \frac{2}{3}\acoss{\frac{-2}{3}} - \frac{(\tau-1)}{3}\]
        \[\IT_3 = t_* + \sinn{\frac{3(t_*-1)+(\tau-1)}{2}}\]
        Solving these equations numerically gives $\tau \approx 1.2319$, $t_* \approx 2.4564$, and $\IT_3 \approx 3.2017$.
    \end{theorem}
    As before, $\tau$ represents the beginning of an interval of time during which the queen cannot be located on the perimeter. In this case, however, $t_*$ is not the first time at which it is possible for the queen to return to the perimeter. Instead it represents a particularly critical time of any algorithm with $n=3$ at which the evacuation time is maximized (although it will happen that $t_*$ occurs before the queen can return to the perimeter). We will show that the queen must explore a subset of the perimeter with total length more than $\tau-1$ before the time $t_*$ in order to evacuate in time less than $\IT_2$.
    
    We begin with a lemma that was first introduced in \cite{CGGKMP}:
    \begin{lemma}\label{lm:chord}
        Consider a perimeter of a disk whose subset of total length $u + \epsilon > 0$ has not been explored for some $\epsilon > 0$ and $\pi \geq u > 0$. Then there exist two unexplored boundary points between which the distance along the perimeter is at least $u$.
    \end{lemma}

    This next lemma is used to determine the critical time $t_*$.
    \begin{lemma}\label{lm:lb3_0}
        Consider an evacuation algorithm with $n$ servants and assume that at the time $t$ the queen has explored a total subset of the perimeter of length $y$. Then, for $x$ and $y$ satisfying $1+\frac{\pi-y}{n} \leq t \leq 1+\frac{2\pi-y}{n}$, it takes time at least $\IT = t+\sinn{\frac{n(t-1)+y}{2}}$ to evacuate the queen.
    \end{lemma}
    \begin{proof}
        Consider an algorithm with evacuation time $\IT$ and with $n$ servants. Then, at the time $t$, the total length of perimeter that the robots have explored is at most $n(t-1) + y \geq \pi$ (since each robot may search at a maximum speed of one, the queen has explored a subset of length $y$, and the robots need at least a unit of time to reach the perimeter). Thus, by Lemma~\ref{lm:chord}, there exists two unexplored boundary points between which the distance along the perimeter is at least $2\pi-n(t-1)-y-\epsilon$ for any $\eps>0$. The chord connecting these points has length at least $2\sinn{\pi - \frac{n(t-1)+y}{2}-\frac{\eps}{2}}$ and an adversary may place the exit at either endpoint of this chord. The queen will therefore take at least $\sinn{\pi - \frac{n(t-1)+y}{2}-\frac{\eps}{2}}$ more time to evacuate and the total evacuation time will be at least $t+\sinn{\pi - \frac{n(t-1)+y}{2}-\frac{\eps}{2}}$. As this is true for any $\eps > 0$ taking the limit $\eps \rightarrow 0$ we obtain $\IT \geq t+\sinn{\pi-\frac{n(t-1)+y}{2}} = t+\sinn{\frac{n(t-1)+y}{2}}$.
    \end{proof}

    In the next two lemmas we show that in order to evacuate in time $\IT < \IT_2$ the queen must explore a length of the perimeter greater than $\tau-1$ and then demonstrate that this is not possible.
    \begin{lemma}\label{lm:lb3_1}
        Define $\tau$ and $t_*$ as in Theorem~\ref{thm:lb_n3}. Then, for $n = 3$ and any evacuation algorithm with $\IT < \IT_3$, the queen must explore a subset of $\UC$ with total length $y > \tau-1$ before the time $t_*$.
    \end{lemma}
    \begin{proof}
        Consider an algorithm with evacuation time $\IT < \IT_3$. We make the assumption that the queen has only explored a subset of total length $y < \tau-1$ at the time $t_*$ and show that this leads to a contradiction.

        Observe that $t_*$ satisfies $1 + \frac{\pi-y}{3} \leq t_* \leq 1 + \frac{2\pi-y}{3}$ for all $y$ satisfying $0 \leq y \leq \tau-1$ and thus, by Lemma~\ref{lm:lb3_0}, we can write
        \[\IT \geq t_* + \sinn{\frac{3(t_*-1) + y}{2}}.\]
        Since $\IT < \IT_3$ we also have
        \[\IT_3 > t_* + \sinn{\frac{3(t_*-1) + y}{2}}.\]
        Since $\IT_3 = t_* + \sinn{\frac{3(t_*-1)+(\tau-1)}{2}}$ we further have
        \[\sinn{\frac{3(t_*-1)+(\tau-1)}{2}} > \sinn{\frac{3(t_*-1)+y}{2}}.\]
        Finally, since $t_* \geq 1 + \frac{\pi-y}{3}$ we know that $\sinn{\frac{3(t_*-1)+y}{2}}$ is a decreasing function of its argument and thus we get
        \[\frac{3(t_*-1)+(\tau-1)}{2} < \frac{3(t_*-1)+y}{2}\]
        which implies that $y > \tau-1$ which contradicts with our assumption that $y < \tau-1$.
    \end{proof}

    \begin{lemma}\label{lm:lb3_2}
        Define $\tau$ and $t_*$ as in Theorem~\ref{thm:lb_n3}. Then, for $n = 3$ and any evacuation algorithm with $\IT < \IT_3$, the queen cannot explore a subset of the perimeter with length $y > \tau-1$ before the time $t_*$.
    \end{lemma}
    \begin{proof}
        As was the case for $n=2$, if the queen is to be on the perimeter at the time $t$ then all of the arc $\theta(t,\IT) = 4\acoss{\frac{\IT-t}{2}}$ must be explored. Since we have $4$ robots in total, the maximum length of arc that can be explored at any time $t$ is $4(t-1)$. However, we can again say that the queen cannot search any of the arc $\theta(t)$ if $t \leq \frac 12 (\IT+1)$. We must therefore have $\theta(t,\IT) \leq 3(t-1)$ for times $t$ that satisfy $t < \frac 12(\IT+1)$.

        Assume first that $\IT \geq 3$. We make the following claim: if $3 \leq \IT < \IT_3$ then the smallest time $t_0 > 0$ solving $\theta(t_0, \IT) = 3(t_0-1)$ satisfies $\left.\diff{\theta}{t} \right|_{t=t_0} > 3$ and $t_0 < \frac 12(\IT+1)$. If this is the case the queen will have to leave the perimeter at the time $t_0$.

        We first demonstrate that $t_0 < \frac 12 (\IT+1)$. Let us rearrange the equation $\theta(t_0, \IT) = 3(t_0-1)$ to get
        \[t_0 = \IT - 2\coss{\frac{3}{4}(t_0-1)}\]
        which is the definition of $\tau$ in Theorem~\ref{thm:lb_n3} (in the case that $\IT = \IT_3$). One can easily confirm that in the case of $\IT = \IT_3$, both $\left.\diff{\theta}{t} \right|_{t=\tau} > 3$ and $\tau < \frac 12(\IT+1)$. Now observe that $\theta(t,\IT)$ is a decreasing function of $\IT$ and this implies that for $\IT < \IT_3$ we have $\theta(\tau,\IT) > \theta(\tau,\IT_3)$. The time $t_0$ must therefore occur earlier than the time $\tau$. We note that $\tau < 2$ and, since we are assuming that $\IT \geq 3$, we have $\tau < \frac 12 (\IT+1)$. Since $t_0 < \tau$ we can finally conclude that $t_0 < \frac 12 (\IT+1)$. 
       
        The second part of the claim follows from Lemma~\ref{lm:dtheta} if we can show that $t_0 < \IT-\frac 23 \sqrt{5}$. We note that $\IT \geq 3$ and thus $\IT-\frac 23 \sqrt{5} \geq 1.5093$. Since $\tau \approx 1.2319$ and $t_0 < \tau$ we can clearly see that $t_0 < \IT-\frac 23 \sqrt{5}$. 
        
        If $\IT < 3$ then it should be obvious that the queen cannot even be at the perimeter at the time $t=1$. Thus, in this case, we take $t_0 = 1$.

        Since the queen must leave the perimeter at the time $t_0 < \tau$, by Lemma~\ref{lm:lb3_1}, we know that the queen must be able to return to the perimeter and explore before the time $t_*$. We claim that this is not possible. Indeed, observe that the queen cannot return to the perimeter until the earliest time $t > t_0$ at which $\theta(t) = 3(t-1)+y$ (where we have set $y < \tau-1$ as the length of the arc $\theta(t)$ explored by the queen). Thus, in order for the queen to have returned to the perimeter before the time $t_*$ we must have $\theta(t_*) \leq 3(t-1)+y$. However, since $\IT < \IT_3$ we have
        $\theta(t_*) = 4\acoss{\frac{\IT-t_*}{2}} > 4 \acoss{\frac{\IT_3 - t_*}{2}}$.
        We note that $\IT_3 - t_* = \sinn{\frac{3(t_*-1)+(\tau-1)}{2}} = \sinn{\acoss{\frac {-2}{3}}} = \sqrt{\frac 59}$ and thus
        $\theta(t_*) > 4 \acoss{\frac{\sqrt{5}}{6}} \approx 4.7556$.
        Since $\tau \approx 1.2319$, and $t_* \approx 2.4564$ we have $3(t_*-1)+y \leq 3(t_*-1) + (\tau-1) \approx 4.6010$. We can therefore see that it is not the case that $\theta(t_*) \leq 3(t-1)+y$ and thus the queen cannot have returned to the perimeter before the time $t_*$. We can finally conclude that the queen can only explore a subset of the perimeter of length $t_0-1 < \tau-1$ before the time $t_*$.
    \end{proof}

        At this point the proof of Theorem~\ref{thm:lb_n3} is trivial.

    \begin{proof}(Theorem~\ref{thm:lb_n3})\\
        Assume we have an algorithm with evacuation time $\IT < \IT_3$. Then, by Lemma~\ref{lm:lb3_1}, the queen must explore a subset of the perimeter of length at least $\tau-1$ by the time $t_*$. However, by Lemma~\ref{lm:lb3_2}, the queen can only explore a subset of the perimeter of length $y < \tau-1$ if $\IT < \IT_3$. We must therefore conclude that it is not possible for the queen to evacuate in time less than $\IT_3$.
    \end{proof}

\section{Conclusion}
\label{secconclusion}

We considered an evacuation problem concerning priority searching on the perimeter of a unit disk where only one robot (the queen) needs to find the exit. In addition to the queen, there are $n \leq 3$ other robots (servants) aiding the queen by contributing to the exploration of the disk but which do not need to evacuate. We proposed evacuation algorithms and studied non-trivial tradeoffs on the queen evacuation time depending on the number $n$ of servants. In addition to analyzing tradeoffs and improving the bounds obtained for the wireless communication model, an interesting open problem would be to investigate other models with limited communication range, e.g., face-to-face.

\bibliographystyle{plain}
\bibliography{refs}

%%
%% Bibliography
%%

%% Please use bibtex, 

%\bibliography{lipics-v2018-sample-article}

\end{document}